\let\oldvec\vec 
\documentclass{llncs}
\let\vec\oldvec 

\usepackage{geometry}
\usepackage{fancyhdr}
\usepackage{amsmath,amssymb}
\usepackage{mathtools}
\usepackage{enumerate}

\usepackage{mathptmx}       
\DeclareMathAlphabet{\mathcal}{OMS}{cmsy}{m}{n}

\usepackage{listings}

\usepackage{booktabs}

\usepackage{tikz}
\usetikzlibrary{arrows,automata,positioning}
\usepackage{pgfplots}

\usetikzlibrary{positioning,chains,fit,shapes,calc}

\newcommand{\N}{\mathbb{N}}

\newcommand{\R}{\mathbb{R}}

\definecolor{myblue}{RGB}{80,80,160}
\definecolor{mygreen}{RGB}{80,160,80}

\usepackage{ifthen}
% ...
\newenvironment{rtheorem}[3][]{%
\noindent\ifthenelse{\equal{#1}{}}{\bf #2 #3.}{\bf #2 #3 (#1)}%
\begin{it}}{\end{it}}

\usepackage{mathptmx}       
\DeclareMathAlphabet{\mathcal}{OMS}{cmsy}{m}{n}

\usepackage{array}
\newcolumntype{x}[1]{>{\centering\arraybackslash\hspace{0pt}}p{#1}}

\title{Tight Inefficiency Bounds for Perception-Parameterized \newline 
Affine Congestion Games}

\author{Pieter Kleer\inst{1} \and Guido Sch\"afer\inst{1,2}}
\authorrunning{P. Kleer \and G.\ Sch\"afer} 

\institute{%
  Centrum Wiskunde \& Informatica (CWI), Networks and Optimization Group \\
  \and
  Vrije Universiteit Amsterdam, Department of Econometrics and Operations Research \\ Amsterdam, The Netherlands \\ 
  \email{kleer@cwi.nl, schaefer@cwi.nl}
}

\begin{document}

\maketitle

\begin{abstract}
Congestion games constitute an important class of non-cooperative games which was introduced by Rosenthal in 1973. In recent years, several extensions of these games were proposed to incorporate aspects that are not captured by the standard model. Examples of such extensions include the incorporation of risk sensitive players, the modeling of altruistic player behavior and the imposition of taxes on the resources. These extensions were studied intensively with the goal to obtain a precise understanding of the inefficiency of equilibria of these games. 

In this paper, we introduce a new model of congestion games that captures these extensions (and additional ones) in a unifying way. The key idea here is to parameterize both the perceived cost of each player and the social cost function of the system designer. Intuitively, each player perceives the load induced by the other players by an extent of $\rho \ge 0$, while the system designer estimates that each player perceives the load of all others by an extent of $\sigma \ge 0$.
The above mentioned extensions reduce to special cases of our model by choosing the parameters $\rho$ and $\sigma$ accordingly. As in most related works, we concentrate on congestion games with affine latency functions here. 

Despite the fact that we deal with a more general class of congestion games, we manage to derive tight bounds on the price of anarchy and the price of stability for a large range of parameters. Our bounds provide a complete picture of the inefficiency of equilibria for these perception-parameterized congestion games. As a result, we obtain tight bounds on the price of anarchy and the price of stability for the above mentioned extensions. Our results also reveal how one should ``design'' the cost functions of the players in order to reduce the price of anarchy. Somewhat counterintuitively, if each player cares about all other players to the extent of $0.625$ only (instead of $1$ in the standard setting) the price of anarchy reduces from $2.5$ to $2.155$ and this is best possible.
\end{abstract}

\newpage
\setcounter{page}{1}

\section{Introduction}\label{sec:intro}

Congestion games constitute an important class of non-cooperative games which was introduced by Rosenthal in 1973 \cite{Rosenthal1973}. In a congestion game, we are given a set of resource from which a set of players can choose. Each resource is associated with a cost function which specifies the cost of this resource depending  on the total number of players using it. Every player chooses a subset of resources (from a set of resource subsets available to her) and experiences a cost equal to the sum of the costs of the chosen resources. 
Congestion games are both theoretically appealing and practically relevant. For example, they have applications in network routing, resource allocation and scheduling problems.

Rosenthal \cite{Rosenthal1973} proved that every congestion game has a pure Nash equilibrium, i.e., a strategy profile such that no player can decrease her cost by unilaterally deviating to another feasible set of resources. 
This result was established through the use of an exact potential function (known as \emph{Rosenthal potential}) satisfying that the cost difference induced by a unilateral player deviation is equal to the potential difference of the respective strategy profiles. In fact, Monderer and Shapley \cite{MS96} showed that the class of games admitting an exact potential function is isomorphic to the class of congestion games. 

One of the main research directions in algorithmic game theory focusses on quantifying the inefficiency caused by selfish behavior. The idea is to assess the quality of a Nash equilibrium relative to an optimal outcome. Here the quality of an outcome is measured in terms of a given \emph{social cost} objective (e.g., the sum of the costs of all players). Koutsoupias and Papadimitriou \cite{Koutsoupias1999} introduced the \textit{price of anarchy} as the ratio between the worst social cost of a Nash equilibrium and the social cost of an optimum. Anshelevich et al.~\cite{Anshelevich2004} defined the \textit{price of stability} as the ratio between the best social cost of a Nash equilibrium and the social cost of an optimum. 


In recent years, several extensions of Rosenthal's congestion games were proposed to incorporate aspects which are not captured by the standard model. For example, these extensions include risk sensitivity of players in uncertain settings \cite{Piliouras2013}, altruistic player behavior \cite{Caragiannis2010Altruism,Chen2014} and congestion games with taxes \cite{Caragiannis2010,Singh}. We elaborate in more detail on these extensions in Section~\ref{sec:prelim}. These games were studied intensively with the goal to obtain a precise understanding of the price of anarchy.

In this paper, we introduce a new model of congestion games, which we term \emph{perception-parameterized congestion games}, that captures all these extensions (and more) in a unifying way. The key idea here is to parameterize both the perceived cost of each player and the social cost function. Intuitively, each player perceives the load induced by the other players by an extent of $\rho \ge 0$, while the system designer estimates that each player perceives the load of all others by an extent of $\sigma \ge 0$.
The above mentioned extensions reduce to special cases of our model by choosing the parameters $\rho$ and $\sigma$ accordingly. 

Despite the fact that we deal with a more general class of congestion games, we manage to derive tight bounds on the price of anarchy and the price of stability for a large range of parameters. Our bounds provide a complete picture of the inefficiency of equilibria for these perception-parameterized congestion games. As a consequence, we obtain tight bounds on the price of anarchy and the price of stability for the above mentioned extensions. While the price of anarchy bounds are (mostly) known from previous results, the price of stability results are new. 
As in \cite{Caragiannis2010,Caragiannis2010Altruism,Chen2014,Piliouras2013,Singh}, we focus here on congestion games with affine cost functions.\footnote{Conceptually, our model can easily be adapted to more general latency functions; only the derivation of the respective bounds is analytically much more involved.}

We illustrate our model by means of a simple example; formal definitions of our perception-parameterized congestion games are given in Section~\ref{sec:prelim}. Suppose we are given a set of $m$ resources and that every player has to choose precisely one of these resources. The cost of a resource $e \in [m]$\footnote{Given a positive integer $m$, we use $[m]$ to refer to the set $\{1, \dots, m\}$.} is given by a cost function $c_e$ that maps the \emph{load} on $e$ to a real value. In the classical setting, the load of a resource $e$ is defined as the total number of players $x_e$ using $e$. That is, the cost that player $i$ experiences when choosing resource $e$ is $c_e(x_e)$. In contrast, in our setting players have different perceptions of the load induced by the other players. More precisely, the \emph{perceived load} of player $i$ choosing resource $e$ is $1+ \rho(x_e - 1)$, where $\rho \ge 0$ is a parameter. Consequently, the perceived cost of player $i$ for choosing $e$ is $c_e(1+ \rho(x_e - 1))$. Note that as $\rho$ increases players care more about the presence of  other players.\footnote{In general, the parameter $\rho$ can be player- or resource-specific, but in this work we concentrate on the homogenous case (i.e., all players have the same perception parameter).} In addition, we introduce a similar parameter $\sigma \ge 0$ for the social cost objective. Intuitively, this can be seen as the system designer's estimate of how each player perceives the load of the other players. In our example, the social cost is defined as $\sum_{e \in [m]} c_e(1+\sigma(x_e-1)) x_e$.

\paragraph{Our Results.} In Section \ref{sec:poa}, we prove a bound of
\begin{equation}\label{eq:intro_poa}
\max\left\{\rho + 1, \frac{2\rho(1 + \sigma) + 1}{\rho + 1} \right\}
\end{equation}
on the price of anarchy of affine congestion games for a large range or parameters $(\rho,\sigma)$ (see also Figure~\ref{fig:c1} for an illustration). For general affine congestion games we prove that this bound is tight. Further, even for the special case of symmetric network congestion games we show that the bound $(2\rho(1+\sigma)+1)/(\rho + 1)$ is asymptotically tight (on the range for which it is attained). 

In Section \ref{sec:pos}, we give a bound of
\begin{equation}\label{eq:intro_pos}
\frac{\sqrt{\sigma(\sigma+2)} + \sigma}{\sqrt{\sigma(\sigma+2)} + \rho - \sigma} 
\end{equation}
on the price of stability of affine congestion games for a large range or pairs $(\rho,\sigma)$ (see below for details). For general affine congestion games we show that this bound is tight. For symmetric network congestion games we give a better (tight) bound on the price of stability for the case $\sigma = 1$ and $\rho \geq 0$ arbitrary (details are given in Section \ref{sec:misc}).

\begin{table}[t]
\centering
\begin{tabular}{|x{5.2cm}|x{2.6cm}|x{3.2cm}|x{1.8cm}|}
\hline
\textbf{Model} &  \textbf{Parameters} & \textbf{PoA} & \textbf{Reference} \\ 
\hline
\hline
Classical & $\rho = \sigma = 1$ & 5/2 & \cite{Christodoulou2005}  \\
Altruism (1)  & $\sigma = 1$, $1 \leq \rho \leq  2$  & $(4\rho+1)/(1+\rho)$ & \cite{Caragiannis2010Altruism,Chen2014} \\
Altruism (2)  & $\sigma = 1$, $2 \leq \rho \leq  \infty$  & $\rho + 1$ & \cite{Chen2014} \\
Uncertainty: risk neutral-players & $\sigma = \rho = 1/2$ & 5/3 &  \cite{Piliouras2013} \\
Uncertainty: Wald's minimax principle & $\sigma = 1/2$, $\rho = 1$ & 2 & \cite{Caragiannis2012,Piliouras2013} \\
Constant universal taxes & $\sigma = 1$, $\rho = h(1)$. & $2.155$ & \cite{Caragiannis2010}  \\
Generalized affine congestion games & $-$ & $\infty$ & $[^*]$ \\
\hline
\end{tabular}

\medskip

\begin{tabular}{|x{5.2cm}|x{2.6cm}|x{3.2cm}|x{1.8cm}|}
\hline
\textbf{Model} &  \textbf{Parameters} & \textbf{PoS} & \textbf{Reference} \\ 
\hline
\hline
Classical & $\rho = \sigma = 1$ & $1.577$ & \cite{Caragiannis2010}  \\
Altruism (1)  & $\sigma = 1$, $1 \leq \rho \leq  2$  & $(\sqrt{3} + 1)/(\sqrt{3} + \rho - 1)$ & $[^*]$ \\
Uncertainty: risk neutral-players & $\sigma = \rho = 1/2$ & $1.447$ &  $[^*]$ \\
Uncertainty: Wald's minimax principle & $\sigma = 1/2$, $\rho = 1$ & $1$ & $[^*]$ \\
Constant universal taxes & $\sigma = 1$, $\rho = h(1)$ & $2.013$ & $[^*]$  \\
Generalized affine congestion games & $-$ & $2$ & $[^*]$ \\
\hline
\end{tabular}
\medskip
\caption{\label{tab:poa}%
An overview of (tight) price of anarchy and price of stability results for certain values of $\rho$ and $\sigma$. Here $h(1) \approx 0.625$ (see Theorem~\ref{thm:poa_upper} for a formal definition). An asterisk indicates that this result is new.}
\end{table}

An overview of the price of anarchy and the price of stability results that we obtain from \eqref{eq:intro_poa} and \eqref{eq:intro_pos} for several applications known in the literature is given in Table~\ref{tab:poa}. The respective references where these bounds were established first are given in the rightmost column. The connection between these applications and our model is discussed in detail at the end of Section~\ref{sec:prelim}.

In light of the above bounds, we obtain an (almost) complete picture of the inefficiency of equilibria (parameterized by $\rho$ and $\sigma$); for example, see Figure~\ref{fig:sigma1} for an illustration of the price of anarchy if $\sigma = 1$. Note that the price of anarchy \emph{decreases} from $\frac{5}{2}$ for $\rho = 1$ to $2.155$ for $\rho = h(1) \approx 0.625$.\footnote{The price of anarchy for $\rho = h(1)$ was first established by Caragiannis et al. \cite{Caragiannis2010}. However, our bounds reveal that the price of anarchy is infact minimized at $\rho = h(1)$ (see also Figure~\ref{fig:sigma1}).} We refer the reader to Section~\ref{sec:misc} for further remarks and discussions of the results.


\section{Our Model, Applications and Related Work}\label{sec:prelim} 

We first formally introduce our model of congestion games with parameterized perceptions. We then show that our model subsumes several other models that were studied in the literature as special cases. 

\paragraph{Congestion Games.} 

A \emph{congestion game} $\Gamma$ is given by a tuple $(N,E,(\mathcal{S}_i)_{i\in N},(c_e)_{e\in E})$ where $N = [n]$ is the set of players, $E$ the set of resources (or facilities), $\mathcal{S}_i \subseteq 2^{E}$ the set of strategies of player $i$, and $c_e : \R_{\ge 0} \rightarrow \R_{\geq 0}$  the cost function of facility $e$. Given a strategy profile $s = (s_1,\dots,s_n) \in \times_i \mathcal{S}_i$, we define $x_e$ as the number of players using resource $e$, i.e., $x_e = x_e(s) = |\{i \in N : e \in s_i\}|$. 
If $\mathcal{S}_i = \mathcal{S}_j$ for all $i,j \in N$, the game is called \textit{symmetric}. For a given graph $G = (V,E)$, we call $\Gamma$ a \textit{(directed) network congestion game} if for every player $i$ there exist $s_i,t_i \in V$ such that $\mathcal{S}_i$ is the set of all (directed) $(s_i,t_i)$-paths in $G$. 
An \textit{affine congestion game} has cost functions of the form $c_e(x) = a_ex + b_e$ with $a_e,b_e \geq 0$. If $b_e = 0$ for all $e \in E$, the game is called \textit{linear}. 

\paragraph{Our Model.}
We introduce our unifying model of perception-parameterized congestion games with affine latency functions.
For a fixed parameter $\rho \geq 0$, we define the cost of player $i \in N$ by  
\begin{equation}\label{eq:playercost}
C_i^\rho(s) = \sum_{e \in s_i} c_e(1 + \rho(x_e - 1)) = a_e[1 + \rho(x_e-1)] + b_e
\end{equation}
for a given strategy profile $s = (s_1,\dots,s_n)$. For a fixed parameter $\sigma \geq 0$, the social cost of a strategy profile $s$ is given by 
\begin{equation}\label{eq:socialcost}
C^\sigma(s) = \sum_i C_i^\sigma(s) = \sum_{e \in E}x_e(a_e[1 + \sigma(x_e-1)] + b_e).
\end{equation}
 We refer to the case $\rho = \sigma = 1$ as the \textit{classical congestion game} with cost functions $c_e(x) = a_e x + b_e$ for all $e \in E$.

\paragraph{Inefficiency of Equilibria.} 

A strategy profile $s$ is a \emph{Nash equilibrium} if for all players $i \in N$ it holds that $C_i^\rho(s) \leq C_i^\rho(s_i',s_{-i})$ for all $s_i' \in \mathcal{S}_i$, where $(s_i',s_{-i})$ denotes the strategy profile in which player $i$ plays $s_i'$ and all the other players their strategy in $s$. The price of anarchy (PoA) and price of stability (PoS) of a game $\Gamma$ are defined as
$$
\text{PoA}(\Gamma,\rho,\sigma) = \frac{\max_{s \in NE} C^\sigma(s)}{\min_{s^* \in \times_i \mathcal{S}_i}C^\sigma(s^*)} \ \ \ \text{and} \ \ \  \text{PoS}(\Gamma,\rho,\sigma) = \frac{\min_{s \in NE} C^\sigma(s)}{\min_{s^* \in \times_i \mathcal{S}_i}C^\sigma(s^*)},
$$
where NE = NE$(\rho)$ denotes the set of Nash equilibria with respect to the player costs as defined in (\ref{eq:playercost}). For a collection of games $\mathcal{H}$ we define $\text{PoA}(\mathcal{H},\rho,\sigma) = \sup_{\Gamma \in \mathcal{H}} \text{PoA}(\Gamma,\rho,\sigma)$ and  $\text{PoS}(\mathcal{H},\rho,\sigma) = \sup_{\Gamma \in \mathcal{H}} \text{PoS}(\Gamma,\rho,\sigma)$.
Rosenthal \cite{Rosenthal1973} shows that classical congestion games (i.e., $\rho = \sigma = 1)$ have an exact potential function: $\Phi : \times_i \mathcal{S}_i \rightarrow \R$ is an \textit{exact potential function} for a congestion game $\Gamma$ if for every strategy profile $s$, for every $i \in N$ and every $s'_i \in \mathcal{S}_i$: \ $\Phi(s) - \Phi(s_{-i},s_i') = C_i(s) - C_i(s_{-i},s_i')$. The Rosenthal potential $\Phi(s) = \sum_{e \in E} \sum_{k=1}^{x_e} c_e(k)$ is an exact potential function for classical congestion games. 

\paragraph{Applications.}

We review various models that fall within, or are related to, the framework proposed above (for certain values of $\rho$ and $\sigma$). These models sometimes interpret the parameters differently than explained above.

\smallskip
\textit{Altruism \cite{Caragiannis2010Altruism,Chen2014}.} We can rewrite the cost of player $i$ as
$$
C^\rho_i(s) = \sum_{e \in s_i} (a_ex_e + b_e) + (\rho - 1)a_e(x_e-1).
$$
The term $(\rho-1)a_e(x_e-1)$ can be interpreted as a ``dynamic'' (meaning load-dependent) tax that players using resource $e$ have to pay.\footnote{In fact, for $\rho = 2$ this corresponds to the \emph{dynamic taxes} as proposed in a technical report by Singh \cite{Singh}.} 
For $1 \leq \rho \leq \infty$ and $\sigma = 1$, this model is equivalent to the altruistic player setting proposed by Caragiannis et al.~\cite{Caragiannis2010Altruism}. Chen et al. \cite{Chen2014} also study this model of altruism for $1 \leq \rho \leq 2$ and $\sigma = 1$. (The equivalence between the altruistic model and our model is proved in Lemma \ref{lem:alt_nash} in the appendix.)

\smallskip
\textit{Constant taxes \cite{Caragiannis2010}.} We can rewrite the cost of player $i$ as 
$$
C^\rho_i(s) = \sum_{e \in s_i} \rho a_e x_e + (1-\rho)a_e + b_e.
$$
Now a strategy profile $s$ is a Nash equilibrium if for every player $i$ and every $s'_i \in \mathcal{S}_i$
$$
\sum_{e \in s_i} \rho a_e x_e + (1-\rho)a_e + b_e \leq \sum_{e \in s_i \cap s_i'} \rho a_e x_e + (1-\rho)a_e + b_e + \sum_{e \in s_i'\setminus s_i} \rho a_e (x_e+1) + (1-\rho)a_e + b_e.
$$
Dividing by $\rho$ gives that $s$ is also a Nash equilibrium for the cost functions 
$$
T^\rho_i(s) = \sum_{e \in s_i}  a_e x_e  + \frac{b_e}{\rho} + \sum_{e \in s_i} \frac{1 - \rho}{\rho} a_e.
$$
That is, $s$ is a Nash equilibrium in a classical congestion game in which players take into account constant resource taxes of the form $(1-\rho)/\rho  \cdot a_e$. Caragiannis, Kaklamanis and Kanellopoulos \cite{Caragiannis2010} study this type of taxes, which they call \textit{universal tax functions}, for $\rho$ satisfying $(1-\rho)/\rho = 3/2\sqrt{3} - 2$. They consider these taxes to be \emph{refundable}, i.e., they are not taken into account in the social cost, which is equivalent to the case $\sigma = 1$. Note that the function $\tau: (0,1] \rightarrow [0,\infty)$ defined by $\tau(\rho) = (1-\rho)/\rho$ is bijective.\footnote{This relation between altruism (or spite) and constant taxes is also mentioned by Caragiannis et al. \cite{Caragiannis2010Altruism}.}

\smallskip
\textit{Risk sensitivity under uncertainty \cite{Piliouras2013}.} Nikolova, Piliouras and Shamma \cite{Piliouras2013} consider congestion games in which there is a (non-deterministic) order of the players on every resource. A player is only affected by players in front of her. That is, the load on resource $e$ for player $i$ in a strict ordering $r$, where $r_e(i)$ denotes the position of player $i$, is given by $x_e(i) = |\{j \in N : r_e(j) \leq r_e(i)\}|$. The cost of player $i$ is then $C_i(s) = \sum_{e \in s_i} c_e(x_e(i))$. Note that $x_e(i)$ is a random variable if the ordering is non-deterministic. The social cost of the model is defined by the sum of all player costs, 
$$
C^{\frac{1}{2}}(s) = \sum_{e \in E} a_e \frac{x_e(x_e+1)}{2} + b_e
$$
which is independent of the ordering $r$. (This holds because in every ordering there is always one player first, one player second, and so on.) Note that the social cost corresponds to the case $\sigma = \frac{1}{2}$ in our framework. Nikolova et al. \cite{Piliouras2013} study various risk attitudes towards the ordering $r$ that is assumed to have a uniform distribution over all possible orderings. The two relevant attitudes are that of \textit{risk-neutral} players and players applying \emph{Wald's minimax principle}. Risk-neutral players define their cost as the expected cost under the ordering $r$, which correspond to the case $\rho = \frac{1}{2}$ in (\ref{eq:playercost}). This can roughly be interpreted as that players expect to be scheduled in the middle on average. Wald's minimax principle implies that players assume a worst-case scenario, i.e., being scheduled last on all the resources. This corresponds to the case $\rho = 1$.

\smallskip
\textit{Approximate Nash equilibria \cite{Christodoulou2011}.} Suppose that $s$ is a Nash equilibrium under the cost functions defined in (\ref{eq:playercost}). Then, in particular, we have
$$
C^1_i(s) \leq C^\rho_i(s) \leq C^\rho_i(s_i', s_{-i}) \leq \rho C^1_i(s_i', s_{-i})
$$
for any player $i$ and $s_i' \in \mathcal{S}_i$ and $\rho \geq 1$. That is, we have $C^1_i(s) \leq \rho\cdot C^1_i(s_i', s_{-i})$ which means that the strategy profile $s$ is a $\rho$-approximate equilibrium, as studied by Christodoulou, Koutsoupias and Spirakis \cite{Christodoulou2011}. In particular, this implies that any upper bound on the price of anarchy, or price of stability, in our framework yields an upper bound on the price of stability for $\rho$-approximate equilibria for the same class of games.


\smallskip
\textit{Generalized affine congestion games.} Let $\mathcal{A}'$ denote the class of all congestion games $\Gamma$ for which all resources have the same cost function $c(x) = a x + b$, where  $a = a(\Gamma)$ and $b = b(\Gamma)$ satisfy $a \geq 0$ and $a+b > 0$. The class of affine congestion games with non-negative coefficients is contained in $\mathcal{A}'$ since every such game can always be transformed\footnote{This transformation can be done in such a way that both PoA and PoS of the game do not change. For a proof the reader is referred to, e.g., \cite[Lemma 4.3]{Chen2014}.} into a game $\Gamma'$ with $a_e = 1$ and $b_e = 0$ for all resources $e \in E'$, where $E'$ is the resource set of $\Gamma'$.
Without loss of generality  we can assume that $a + b = 1$, since the cost functions can be scaled by $1/(a+b)$. The cost functions of  $\Gamma \in \mathcal{A}'$ can then equivalently be written as $c(x) = \rho x + (1 - \rho)$ for $\rho \geq 0$. This is precisely the definition of $C^\rho_i(s)$ (with $a_e = 1$ and $b_e = 0$ taken there). In particular, if we take $\sigma = \rho$, meaning that $ C^\rho(s) = \sum_{i \in N} C^\rho_i(s)$, we have
$$
\text{PoA}(\mathcal{A}') = \sup_{\rho \geq 0} \ \text{PoA}(\mathcal{A},\rho,\rho)
\quad\text{and}\quad
\text{PoS}(\mathcal{A}') = \sup_{\rho \geq 0} \ \text{PoS}(\mathcal{A},\rho,\rho),
$$
where $\mathcal{A}$ denotes the class of affine congestion games with non-negative coefficients.

\smallskip\noindent
Due to page limitations some material is omitted from the main text below and can be found in the appendix.

\section{Price of Anarchy}\label{sec:poa}

In this section we derive an upper bound on the price of anarchy for affine congestion games of
$$
\max\left\{\rho + 1, \frac{2\rho(1 + \sigma) + 1}{\rho + 1} \right\}
$$
for a wide range of pairs $(\rho,\sigma)$ that captures all known (to us) price of anarchy results in the literature that fall within our model. We show that this bound is tight for general congestion games. We also show that the bound $(2\rho(1 + \sigma) + 1)/(\rho + 1)$ is (asymptotically) tight for symmetric network congestion games, for the range of $(\rho,\sigma)$ on which it is attained. 

\begin{figure}[t]
\centering
\scalebox{0.8}{
\begin{tikzpicture}

\draw[->] (0,0) -- (8.5,0) node[anchor=north] {$\sigma$};
\draw	(0,0) node[anchor=north] {0}
		(2,0) node[anchor=north] {$\frac{1}{2}$}
		(4,0) node[anchor=north] {$1$};

\draw[->] (0,0) -- (0,6) node[anchor=east] {$\rho$};
\draw[very thick] (2,0) -- (2,6);

\draw[dotted] (2,1) -- (8,4);
\draw (9,4) node {$\rho = \sigma$};
\draw[very thick] (2,2) -- (6,6);
\draw (7,6) node {$\rho = 2\sigma$};

\draw (2,0.5) edge[bend left=10, very thick] (8,2.75);   
\draw (9,2.75) node {$\rho = h(\sigma)$};

\draw (3,5) node {$\rho + 1$}; 
\draw (6,4) node {$\dfrac{2\rho(1 + \sigma) + 1}{\rho + 1}$}; 
\draw (6,1) node {$?$}; 

\draw[dotted] (2.5,0.8) -- (2.5,1.25);
\draw[dotted] (3,1.1) -- (3,1.5);
\draw[dotted] (3.5,1.3) -- (3.5,1.75);
\draw[dotted] (4,1.55) -- (4,2);
\draw[dotted] (4.5,1.75) -- (4.5,2.25);
\draw[dotted] (5,2) -- (5,2.5);
\draw[dotted] (5.5,2.2) -- (5.5,2.75);
\draw[dotted] (6,2.4) -- (6,3);
\draw[dotted] (6.5,2.5) -- (6.5,3.25);
\draw[dotted] (7,2.6) -- (7,3.5);
\draw[dotted] (7.5,2.7) -- (7.5,3.75);
\draw[dotted] (8,2.8) -- (8,4);


\end{tikzpicture}}
\caption{The bound $\rho + 1$ holds for $\rho \geq 2\sigma \geq 1$. The bound $(2\rho(1+\sigma)+1)/(1 + \rho)$ holds for $\sigma \leq \rho \leq 2\sigma$. Roughly speaking, this bound also holds for $h(\sigma) \leq \rho \leq \sigma$, but our proof of Theorem \ref{thm:poa_upper} only works for a discretized range of $\sigma$ (hence the vertical dotted lines in this area). The function $h$ is given in Theorem \ref{thm:poa_upper}.} 
\label{fig:c1}
\end{figure}
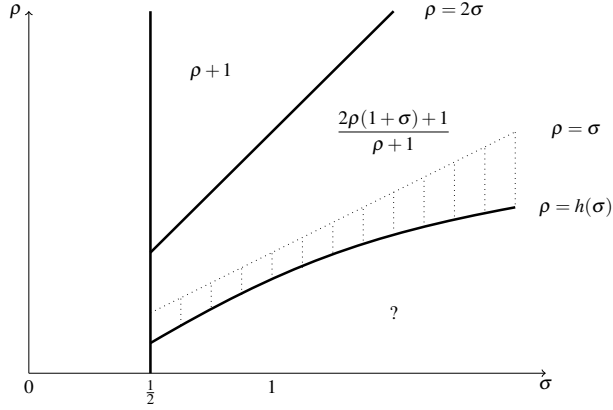
We need the following technical result for the proof of the main result in Theorem \ref{thm:poa_upper}: 

\begin{lemma} \label{lem:main_ineq}
Let $s$ be a Nash equilibrium under the cost functions ${C}_i^\rho(s)$ and let $s^*$ be a minimizer of $C^\sigma(\cdot)$. For $\rho, \sigma \geq 0$ fixed, if there exist $\alpha(\rho,\sigma), \beta(\rho,\sigma) \geq 0$ such that
$$
(1 + \rho\cdot x)y - \rho(x - 1)x - x \leq - \beta(\rho,\sigma)(1 + \sigma(x - 1))x + \alpha(\rho,\sigma)(1 + \sigma(y - 1))y
$$
for all non-negative integers $x$ and $y$, then $\beta(\rho,\sigma) C^\sigma(s) \le \alpha(\rho,\sigma) C^\sigma(s^*)$.
\end{lemma}

\begin{theorem}\label{thm:poa_upper}
Let $s$ be a Nash equilibrium under the cost functions ${C}_i^\rho(s)$ and let $s^*$ be a minimizer of $C^\sigma(\cdot)$. Then
\begin{equation}\label{eq:main_poa_result}
\frac{C^\sigma(s)}{C^\sigma(s^*)} \leq \frac{2\rho(1 + \sigma) + 1}{\rho + 1} 
\end{equation}
\begin{enumerate}[i)]
\item if $\frac{1}{2} \leq \sigma \leq \rho \leq 2\sigma$,
\item if $\sigma = 1$ and $h(\sigma) \leq \rho \leq 2 \sigma$,
\end{enumerate}
where $h(\sigma) = g(1 + \sigma + \sqrt{\sigma(\sigma + 2)}, \sigma)$ is the optimum of the function
$$
g(a,\sigma) = \frac{\sigma(a^2 - 1)}{(1 + \sigma)a^2 - (2\sigma + 1)a + 2\sigma(\sigma+1)}.
$$
Furthermore, there exists a function $\Delta = \Delta(\sigma)$ (specified in the appendix) satisfying for every fixed $\sigma_0 \geq 1/2$: if $\Delta(\sigma_0) \geq 0$, then (\ref{eq:main_poa_result}) is true for all $h(\sigma_0) \leq \rho \leq 2\sigma_0$. 
\end{theorem}
\begin{proof}[Sketch] For the functions $\alpha(\rho,\sigma) = (2\rho(1+\sigma)+1)/(1+2\sigma)$ and $\beta(\rho,\sigma) = (1+\rho)/(1+2\sigma)$, we prove the inequality in Lemma \ref{lem:main_ineq}. We show that, for certain functions $f_1$ and $f_2$, the smallest $\rho$ satisfying the inequality of Lemma \ref{lem:main_ineq} is given by the quantity
$$
h(\sigma) = \sup_{x,y \in \N: f_1(x,y,\sigma) > 0} - \frac{f_2(x,y,\sigma)}{f_1(x,y,\sigma)}.
$$
We divide the set $(x,y) \in \N \times \N$ in lines of the form $x = ay$ and determine the supremum over every line.  After that we take the supremum over all lines, which then gives the desired result. We first show that the case $x \leq y$ which is trivial. We then focus on $y > x$. In this case we want to determine
$$
h(\sigma) = \sup_{a  \in \R_{> 1}} \sup_{ y > 1} - \frac{f_2(ay,y,\sigma)}{f_1(ay,y,\sigma)}.
$$
We show that $h(\sigma) = \max\{\gamma_1(\sigma),\gamma_2(\sigma)\}$ for certain functions $\gamma_1$ and $\gamma_2$. Numerical experiments suggest that $\Delta(\sigma) := \gamma_1(\sigma) - \gamma_2(\sigma) \geq 0$, that is, the maximum is always attained for $\gamma_1$ (which is the definition of $h$ given in the statement of the theorem). In particular, this means that if, for a fixed $\sigma$, the non-negativity of $\Delta(\sigma)$ is checked, then the proof indeed yields the inequality of Lemma \ref{lem:main_ineq} for $h(\sigma) \leq \rho \leq 2\sigma$. \qed
\end{proof}

Numerical experiments suggest that $\Delta(\sigma)$ is non-negative for all $\sigma \geq 1/2$.
We emphasize that for a fixed $\sigma$, with $\Delta(\sigma) \geq 0$, the proof that the inequality holds for all $h(\sigma) \leq \rho \leq 2\sigma$ is exact in the parameter $\rho$. The first two cases of Theorem \ref{thm:poa_upper} capture all the price of anarchy results from the literature. 

We next show that the bound of Theorem \ref{thm:poa_upper} is also an (asymptotic) lower bound for linear symmetric network congestion games.\footnote{In the appendix (Theorem \ref{thm:poa_lower}) we show (non-asymptotic) tightness for general congestion games.} This improves a result in the risk-uncertainty model of Piliouras et al. \cite{Piliouras2013}, who only prove asymptotic tightness for symmetric linear congestion games (for their respective values of $\rho$ and $\sigma$). It also improves a result in the altruism model by Chen et al. \cite{Chen2014}, who show tightness only for general congestion games. The proof is a generalization of the construction of Correa et al. \cite{Correa2015}, who proved that for the classical case $\rho = \sigma = 1$, the Price of Anarchy upper bound of $5/2$, shown in \cite{Christodoulou2005}, is asymptotically tight for symmetric network congestion games.

\begin{theorem}\label{thm:poa_network_lower}
For $\rho, \sigma > 0$ fixed, and players with cost functions ${C}_i^\rho(s)$, there exist symmetric network linear congestion games such that
$$
\frac{C^\sigma(s)}{C^\sigma(s^*)} \geq \frac{2\rho(1 + \sigma) + 1}{\rho + 1} - \epsilon
$$
for any $\epsilon > 0$, where $s$ is a Nash equilibrium, and $s^*$ a socially optimal strategy profile.
\end{theorem}
\begin{proof}
We construct a symmetric network linear congestion game with  $n$ players. 
We first describe the graph topology used in the proof of Theorem 5 in \cite{Correa2015} (using similar notation and terminology). 
The graph $G$ consists of $n$ principal disjoint paths, called $P_1,\dots,P_n$ from $s$ to $t$ (horizontal paths in Figure \ref{fig:poa_network_lower_sketch}), each consisting $2n-1$ arcs (and hence $2n$ nodes). With $e_{i,j}$ the $j$-th arc on path $i$ is denoted for $i = 1,\dots,n$ and $j = 1,\dots,2n-1$. Also, $v_{i,j}$ denotes the $j$-th node on path $i$ for $i = 1,\dots,n$ and $j = 1,\dots,2n$. There are also $n(n-1)$ connecting arcs: for every path $i$ there is an arc $(v_{i,2k+1}, v_{i-1,2k})$ for $k = 1,\dots,n-1$, where $i - 1$ is taken modulo $n$ (the diagonal arcs in Figure \ref{fig:poa_network_lower_sketch}). For $j \geq 1$ fixed, we say that the arcs $e_{i,j}$ for $i = 1,\dots,n$ form the $(j-1)$-th layer of $G$ (see Figure \ref{fig:poa_network_lower_sketch}).

\begin{figure}[t]
\centering \scalebox{1}{
\begin{tikzpicture}
\begin{scope}


\node at (-2,2) [circle,fill=black, scale = 0.4] (s) {};
\node at (-2.3,2) [] {$s$};
\node at (9.1,2) [circle,fill=black, scale = 0.4] (t) {};
\node at (9.4,2) [] {$t$};
\foreach \y in {0,...,4} 
	\foreach \x in {0,...,7}
		 {\node at (\x,\y)  [circle,fill=black, scale = 0.4] (\x\y) {};} 
       
\foreach \y  in {0,...,4}       
 \foreach \x [count=\xi] in {0,...,6}
		\draw (\x\y) edge[->] (\xi\y);

\node at (-1.5,3.5) [] {$(1+\rho)x$};
\node at (0.5,4.3) [] {$\rho x$};
\node at (1.5,4.3) [] {$x$};
\node at (2.5,4.3) [] {$\rho x$};
\node at (3.5,4.3) [] {$x$};
\node at (4.5,4.3) [] {$\rho x$};
\node at (5.5,4.3) [] {$x$};
\node at (6.5,4.3) [] {$\rho x$};
\node at (8.5,3.5) [] {$(1+\rho)x$};

\draw (-2,2) edge[->] (-0.05,0);
\draw (-2,2) edge[->] (-0.05,1);
\draw (-2,2) edge[->] (-0.05,2);
\draw (-2,2) edge[->] (-0.05,3);
\draw (-2,2) edge[->] (-0.05,4);

\draw (7,0) edge[->] (t);
\draw (7,1) edge[->] (t);
\draw (7,2) edge[->] (t);
\draw (7,3) edge[->] (t);
\draw (7,4) edge[->] (t);

\draw (1,0) edge[->] (0,3.95);
\draw (3,0) edge[->] (2,3.95);
\draw (5,0) edge[->] (4,3.95);
\draw (7,0) edge[->] (6,3.95);

\foreach \y in {1,...,4} 
\draw (1,\y) edge[->] (0.07,\y-0.95);

\foreach \y in {1,...,4} 
\draw (3,\y) edge[->] (2.07,\y-0.95);

\foreach \y in {1,...,4} 
\draw (5,\y) edge[->] (4.07,\y-0.95);

\foreach \y in {1,...,4} 
\draw (7,\y) edge[->] (6.07,\y-0.95);

\draw (-2,2) edge[->,blue, dashed, very thick] (0,3);
\draw (0,3) edge[->,blue, dashed, very thick] (1,3);
\draw (1,3) edge[->,blue, dashed, very thick] (0,2); 
\draw (0,2) edge[->,blue, dashed, very thick] (1,2); 
\draw (1,2) edge[->,blue, dashed, very thick] (2,2); 
\draw (2,2) edge[->,blue, dashed, very thick] (3,2); 
\draw (3,2) edge[->,blue, dashed, very thick] (2,1); 
\draw (2,1) edge[->,blue, dashed, very thick] (3,1); 
\draw (3,1) edge[->,blue, dashed, very thick] (4,1); 
\draw (4,1) edge[->,blue, dashed, very thick] (5,1); 
\draw (5,1) edge[->,blue, dashed, very thick] (4,0); 
\draw (4,0) edge[->,blue, dashed, very thick] (5,0); 
\draw (5,0) edge[->,blue, dashed, very thick] (6,0); 
\draw (6,0) edge[->,blue, dashed, very thick] (7,0); 
\draw (7,0) edge[->,blue, dashed, very thick] (6,4); 
\draw (6,4) edge[->,blue, dashed, very thick] (7,4);  
\draw (7,4) edge[->,blue, dashed, very thick] (9,2);  

\draw (-2,2+0.13) edge[->,red, very thick] (0,1+0.13); 
\draw (0,1+0.13) edge[->,red, very thick] (1,1+0.13); 
\draw (1,1+0.13) edge[->,red,  very thick] (0,0+0.13); 
\draw (0,0+0.13) edge[->,red, very thick] (1,0+0.13); 
\draw (1,0+0.13) edge[->,red, very thick] (2,0+0.13); 
\draw (2,0+0.13) edge[->,red,  very thick] (3,0+0.13);  
\draw (3,0+0.13) edge[->,red,  very thick] (4,0+0.13);  

\draw (1,3+0.13) edge[->,red, very thick] (2,3+0.13); 
\draw (2,3+0.13) edge[->,red, very thick] (3,3+0.13); 
\draw (3,3+0.13) edge[->,red,  very thick] (4,3+0.13); 
\draw (4,3+0.13) edge[->,red, very thick] (5,3+0.13); 
\draw (5,3+0.13) edge[->,red, very thick] (4,2+0.13); 
\draw (4,2+0.13) edge[->,red,  very thick] (5,2+0.13);  
\draw (5,2+0.13) edge[->,red,  very thick] (6,2+0.13);  
\draw (6,2+0.13) edge[->,red,  very thick] (7,2+0.13);  
\draw (7,2+0.13) edge[->,red,  very thick] (6,1+0.13);  
\draw (6,1+0.13) edge[->,red,  very thick] (7,1+0.13);  
\draw (7,1+0.13) edge[->,red,  very thick] (6,0+0.13);  

\node at (-1,-0.3) [] {$0$};
\node at (0.5,-0.3) [] {$1$};
\node at (1.5,-0.3) [] {$2$};
\node at (2.5,-0.3) [] {$3$};
\node at (3.5,-0.3) [] {$4$};
\node at (4.5,-0.3) [] {$5$};
\node at (5.5,-0.3) [] {$6$};
\node at (6.5,-0.3) [] {$7$};
\node at (8,-0.3) [] {$8$};
 
\node at (1,3.3) [] {$v_1$};
\node at (1.9,3.3) [] {$w_1$};
\node at (7.2,0.8) [] {$v_2$};
\node at (5.8,0.2) [] {$w_2$};

\end{scope}
\end{tikzpicture}}
\caption{Illustration of the instance for $n = 5$. The dashed (blue) path indicates the strategy of player $2$ in the Nash equilibrium. For every principal path, the first and last arc have cost $(1+\rho)x$, and in between the costs alternate between $\rho x$ and $x$ (starting and ending with $\rho x$). The diagonal connecting arcs have cost zero. The numbers at the bottom indicate the layers. The bold (red) subpaths indicate the two deviation situations that are analyzed to prove that $s$ is indeed a Nash equilibrium.} 
\label{fig:poa_network_lower_sketch}
\end{figure}
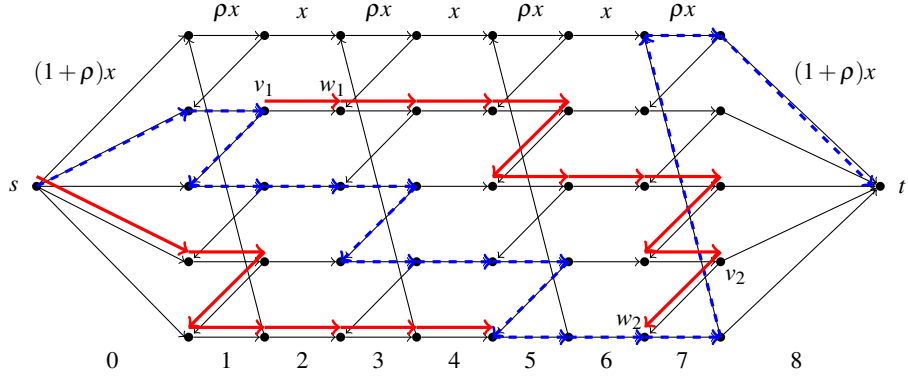

The cost functions are as follows. All arcs leaving $s$ (the arcs $e_{i,1}$ for $i = 1,\dots,n$) and all arcs entering $t$ (the arcs $e_{i,2n}$ for $i = 1,\dots,n$) have latency $c_e(x) = (1+\rho)x$. For all $i = 1,\dots,n$, the arcs $e_{i,2k-1}$ for $k = 1,\dots,n-1$ have cost function $c_e(x) = \rho x$, whereas the arcs $e_{i,2k}$ for $k = 1,\dots,n-2$ have cost function $c_e(x) = x$. All other arcs (the diagonal connecting arcs) have cost zero. 

The feasible strategy profile $t$ in which player $i$ uses principal path $P_i$, for all $i = 1,\dots,n$ has social cost $C^\sigma(t) = n(2(1+\rho) + (n-1)\rho + (n-2)) = n((1+\rho)n +\rho)$. A Nash equilibrium is given by the strategy profile in which every player $k$ uses the following path: she starts with arcs $e_{k,1}$ and $e_{k,2}$, then uses all arcs of the form $e_{k+j,2j}, e_{k+j,2j+1}, e_{k+j,2j+2}$ for $j = 1,\dots,n-1$, and ends with arcs $e_{k+n-1,2n-2}, e_{k+n-1,2n-1}$ (and uses all connecting arcs in between).\footnote{This is similar to the construction in Theorem 5 \cite{Correa2015}.} Note that all the (principal) arcs of layer $j$ have load $1$ is $j$ is even, and load $2$ if $j$ is odd. The social cost of this profile is given by
$C^\sigma(s) = n(2(1+\rho) + (n-1) \cdot 2 \cdot \rho(1+\sigma(2-1)) + n-2) = n((1+2\rho(1+\sigma))n - 2\rho\sigma)$. It follows that $C^\sigma(s)/C^\sigma(t) \uparrow (1+2\rho(1+\sigma))/(1+\rho)$ as $n \rightarrow \infty$. We now show that the above mentioned strategy profile $s$ is indeed a Nash equilibrium. 

Fix some player, say player $2$, as in Figure \ref{fig:poa_network_lower_sketch}, and suppose that this player deviates to some path $Q$. Let $j$ be the first layer in which $P_2$ and $Q$ overlap. Note that $j$ must be odd. The cost $C_2^\rho(s)$ of player $2$, on the subpath of $P_2$ leading to the first overlapping arc with $Q$, is at most
$$
(1+\rho) + \frac{j-1}{2}\cdot [2\cdot[\rho(1 + \rho(2-1))] + 1] + \rho(1 + \rho(2-1)) =(1+\rho)^2 + \frac{j-1}{2}(1 + 2\rho(1+\rho))
$$
The subpath of $Q$ leading to the first overlapping arc with $P_2$ has $C^\rho_i(Q,s_{-i})$ as follows. She uses at least one arc in every odd layer (before the overlapping layer) with a load of $3$ and one arc of every even layer (before the overlapping arc) with load $2$, meaning that the cost of player $i$ on the subpath of $Q$ is at least
$$
(1+\rho)(1+\rho(2 - 1)) + \frac{j-1}{2}\cdot [(\rho(1 + \rho(3-1))) + (1 + \rho(2-1))] = (1+\rho)^2 + \frac{j-1}{2}\cdot (2\rho + 1)(1+\rho)
$$
Since $1 + 2\rho(1+\rho) < (2\rho+1)(1+\rho)$ for all $\rho \geq 0$, if follows that the cost of player $i$ on the subpath of $P_2$ is no worse than that of the subpath of $Q$, when player $2$ deviates from $P_2$ to $Q$. If follows that it suffices to show that $P_2$ is an equilibrium strategy in $s$ with respect to deviations $Q$ that overlap on the first arc $e_{2,1}$ with $P_2$. A similar argument shows that it also suffices to look at deviations $Q$ for which $Q$ and $P_2$ overlap on the last arc $e_{2,2n-1}$ of $P_2$.

Now suppose that $P_2$ and $Q$ do not overlap on some internal part of $P_2$. Note that the first arc of $Q$ that is not contained in $P_2$, say $(v_1,w_1)$ must be in an even layer, and also that the last arc, say $(v_2,w_2)$ (which is a connecting arc) is in an odd layer (note that $v_1 \neq s$ and $w_2 \neq t$ w.l.o.g. by what is said in the previous paragraph). It is not hard to see that the subpath of $Q$ from $v_1$ to $w_2$ contains the same number of even-layered arcs as the subpath of $P_2$, and the same number of odd-layered arcs as the subpath of $P_2$. However, the load on all the odd-layered arcs on the subpath of deviation $Q$ is $3$, whereas the load on odd-layered arcs in the subpath of $P_2$ between $v_1$ and $w_2$ (in strategy $s$) is $2$. Similarly, the load on every even-layered arc on the subpath of deviation $Q$ is $2$, whereas the load on ever even-layered arc in the subpath of $P_2$ is $1$. Hence the subpath of deviation $Q$ between $v_1$ and $w_2$ can never be profitable.\qed
\end{proof}

For $\rho \geq 2\sigma$, we can obtain a tight bound of $\rho + 1$ on the price of anarchy. Remarkably, the bound itself does not depend on $\sigma$, only the range of $\rho$ and $\sigma$ for which it holds. For the parameters $\sigma = 1$ and $\rho \geq 2$ in the altruism model of Caragiannis et al. \cite{Caragiannis2010Altruism}, this bound is known to be tight for non-symmetric singleton congestion games (where all strategies consist of a single resource). We only provide tightness for general congestion games, but the construction is significantly simpler.

\begin{theorem}
Let $s$ be a Nash equilibrium under the cost functions ${C}_i^\rho(s)$ and let $s^*$ be a minimizer of $C^\sigma(\cdot)$. Then
$C^\sigma(s)/C^\sigma(s^*) \leq \rho + 1$ for $1 \leq 2\sigma \leq \rho$. Furthermore, this bound is tight.
\label{thm:above_twosigma}
\end{theorem}

\section{Price of Stability}\label{sec:pos} 
In this section we give a tight bound for the price of stability of 
$$
\frac{\sqrt{\sigma(\sigma+2)} + \sigma}{\sqrt{\sigma(\sigma+2)} + \rho - \sigma} 
$$
for a large range of pairs $(\rho,\sigma)$. 
We show this bound to be (asymptotically) tight. 

We need the following technical lemma.

\begin{lemma}\label{lem:main_ineq_pos}
For all non-negative integers $x$ and $y$, and $\sigma \geq 0$ arbitrary, we have
$$
\left(x - y + \frac{1}{2}\right)^2- \frac{1}{4} + 2\sigma x(x - 1) + (\sqrt{\sigma(\sigma+2)} + \sigma)[y(y-1) - x(x-1)] \geq 0. 
$$
\end{lemma}\medskip

\begin{theorem}\label{thm:pos_upper} Let $s$ be a best Nash equilibrium under the cost functions ${C}_i^\rho(s)$ and let $s^*$ be a minimizer of $C^\sigma(\cdot)$. We have
$$
\frac{C^\sigma(s)}{C^\sigma(s^*)} \leq \frac{\sqrt{\sigma(\sigma+2)} + \sigma}{\sqrt{\sigma(\sigma+2)} + \rho - \sigma} 
$$
for all $\sigma > 0$ and
$$
\frac{2\sigma}{1 + \sigma + \sqrt{\sigma(\sigma+2)}} \leq \rho \leq 2\sigma.
$$
\end{theorem}

An overview of the implications of this bound is given in Table \ref{tab:poa}. 
The bound of $2$ for generalized affine congestion games requires some additional arguments (see appendix).

\begin{proof}
Our proof is similar to a technique of Christodoulou, Koutsoupias and Spirakis \cite{Christodoulou2011} for upper bounding the price of stability of $\rho$-approximate equilibria (we elaborate on the conceptual difference in Section \ref{sec:misc}). However, for a general $\sigma$ the analysis is more involved. The main technical contribution comes from establishing the inequality in Lemma \ref{lem:main_ineq_pos}.

We can write $C^\rho_i(s) = a_e x_e + b_e + (\rho - 1)a_ex_e$, which, by Rosenthal \cite{Rosenthal1973}, implies that
$$
\Phi^\rho(s) := \sum_{e \in E} a_e\frac{x_e(x_e + 1)}{2}  + b_e x_e + (\rho - 1)\sum_{e \in E} a_e\frac{(x_e - 1)x_e}{2} 
$$
is an exact potential for $C^\rho_i(s)$. 

The idea of the proof is to combine the Nash inequalities, and the fact that the global minimum of $\Phi^\rho(\cdot)$ is a Nash equilibrium (because it is an exact potential). Let $s$ denote the global minimum of $\Phi^\rho$, and $s^*$ a socially optimal solution. We can without loss of generality assume that $a_e = 1$ and $b_e = 0$. The Nash inequalities (as in the price of anarchy analysis) yield
$$
\sum_{e\in E} x_e(1 + \rho(x_e - 1)) \leq \sum_{e \in E} (1+\rho x_e)x_e^*
$$
whereas the fact that $s$ is a global optimum of $\Phi^\rho(\cdot)$ yields $\Phi^\rho(s) \leq \Phi^\rho(s^*)$, which reduces to
$$
\sum_{e \in E} \rho x_e^2 + (2 - \rho)x_e \leq \sum_{e \in E} \rho (x_e^*)^2 + (2 - \rho)x_e^*. 
$$
If we can find $\gamma, \delta \geq 0$, and some $K \geq 1$, for which
$$
\big( 0 \leq \big)\ \  \gamma\left[ \rho(x_e^*)^2 + (2 - \rho)x_e^* - \rho x_e^2 - (2 - \rho)x_e \right] + \delta\left[ (1+\rho x_e)x_e^* - x_e(1 + \rho(x_e - 1)\right]
$$
\begin{equation}\label{eq:pos_proof}
\leq K \cdot x_e^*[1 + \sigma(x_e^*-1)] - x_e[1 + \sigma(x_e-1)],
\end{equation}
then this implies that $C^\sigma(s)/C^\sigma(s^*) \leq K$. We take
$$
\delta = \frac{K - 1}{\rho} \ \ \ \ \text{ and } \ \ \ \ \gamma = \frac{(\rho - 1)K + 1}{2\rho}.
$$
It is not hard to see that $\delta \geq 0$ always holds, however, for $\gamma$ we have to be more careful. We will later verify for which combinations of $\rho$ and $\sigma$ the parameter $\gamma$ is indeed non-negative. Rewriting the expression in (\ref{eq:pos_proof}) yields that we have to find $K$ satisfying
$$
K \geq \frac{f_2(x_e,x_e^*,\sigma)}{f_1(x_e,x_e^*,\rho,\sigma)} :=  \frac{(x_e^*)^2 - 2x_ex_e^* + (1+2\sigma)x_e^2 - x_e^* + (1-2\sigma)x_e}{\left[(1 - \rho + 2\sigma)(x_e^*)^2 - 2x_ex_e^* + (1+\rho)x_e^2 + (\rho - 1 - 2\sigma)x_e^* - (\rho - 1)x_e\right]}.
$$
Note that this reasoning is only correct if $f_1(x_e,x_e^*,\rho,\sigma) \geq 0$. This is true since 
$$
f_1(x_e,x_e^*,\rho,\sigma) = \left(x_e - x_e^* + \frac{1}{2}\right)^2 - \frac{1}{4} + (2\sigma - \rho)x_e^*(x_e^*-1) + \rho x_e(x_e-1)
$$
which is non-negative for all $x_e,x_e^* \in \N$, $\sigma \geq 0$ and $0 \leq \rho \leq 2\sigma$. Furthermore, the expression is zero if and only if $(x_e,x_e^*) \in \{(0,1),(1,1)\}$, but for these pairs the nominator is also zero, and hence, the expression in (\ref{eq:pos_proof}) is therefore satisfied for those pairs. We can write
$$
f_2(x_e,x_e^*,\sigma) = \left(x_e - x_e^* + \frac{1}{2}\right)^2- \frac{1}{4} + 2\sigma x_e(x_e - 1)
$$
and therefore $f_2/f_1 = \frac{A}{A + (2\sigma - \rho)B}$, where
$$
A = \left(x_e - x_e^* + \frac{1}{2}\right)^2- \frac{1}{4} + 2\sigma x_e(x_e - 1),  \ \ \ B =  x_e^*(x_e^*-1) - x_e(x_e-1).
$$
Note that if $\rho = 2\sigma$, we have $f_2/f_1 = 1$, and hence we can take $K = 1$. Otherwise,
$$
\frac{A}{A + (2\sigma - \rho)B} \leq \frac{\sqrt{\sigma(\sigma+2)} + \sigma}{\sqrt{\sigma(\sigma+2)} + \rho - \sigma} =: K \ \ \ \Leftrightarrow \ \ \ A + (\sqrt{\sigma(\sigma+2)}+\sigma)B \geq 0
$$
The inequality on the right is true by Lemma \ref{lem:main_ineq_pos}.

To finish the proof, we determine the pairs $(\rho,\sigma)$ for which the parameter $\gamma$ is non-negative. This holds if and only if
$$
(\rho - 1)K + 1 = (\rho - 1)\frac{\sqrt{\sigma(\sigma+2)} + \sigma}{\sqrt{\sigma(\sigma+2)} + \rho - \sigma} + 1 \geq 0.
$$
Rewriting this yields the bound on $\rho$ in the statement of the theorem. \qed
\end{proof}

\bigskip

In the next theorem we provide a lower bound on the price of stability for arbitrary non-negative pairs $(\rho,\sigma)$. The proof is similar to a construction of Christodoulou et al. \cite{Christodoulou2011} used to give a lower bound on the price of stability for $\rho$-approximate equilibria. The key difference is to tune the parameters $\alpha, \beta$ in the proof with respect to the Nash definition based on the cost function $C_i^\rho(\cdot)$, rather than the $\rho$-approximate Nash definition.

\begin{theorem}
Let $\rho, \sigma > 0$ fixed, with $\rho < 2\sigma$, and $\epsilon > 0$ arbitrary. Then there exists a linear congestion game, with player cost functions $C_i^\rho(s)$, such that
$$
\frac{C^\sigma(s)}{C^\sigma(s^*)} \geq \frac{\sqrt{\sigma(\sigma+2)} + \sigma}{\sqrt{\sigma(\sigma+2)} + \rho - \sigma} - \epsilon.
$$
Here, $s$ is a best Nash equilibrium, and $s^*$ a social optimum.
\end{theorem}
\begin{proof}
We describe the construction of Theorem 9 \cite{Christodoulou2011} using similar notation. We have a game of $n = n_1 + n_2$ players divided into two sets $G_1$ and $G_2$ with size resp. $n_1$ and $n_2$. Each player $i \in G_1$ has two strategies: $A_i$ and $P_i$. The players in $G_2$ have a unique strategy $D$. The strategy profile $A = (A_1,\dots,A_{n_1},D,\dots,D)$ will be the unique Nash equilibrium, and the strategy profile $P = (P_1,\dots,P_{n_1},D,\dots,D)$ will be the social optimum. We have three types of resources:
\begin{itemize}
\item $n_1$ resources $\alpha_i$, $i = 1,\dots n_1$, with cost function $c_{\alpha_i}(x) = \alpha x$. The resource $\alpha_i$ only belongs to strategy $P_i$.
\item $n_1(n_1 - 1)$ resources\footnote{The proof of Theorem 9 \cite{Christodoulou2011} contains a typo here: it says there are $n(n-1)$ resources of this type, instead of $n_1(n_1-1)$.} $\beta_{ij}$, $i,j = 1,\dots,n_1$ with $i \neq j$, with cost function $c_{\beta_{ij}}(x) = \beta x$. The resource $\beta_{ij}$ belongs only to strategies $A_i$ and $P_j$.
\item One resource $\gamma$ with cost function $c_{\gamma}(x) = x$, that belongs to $A_i$ for $i = 1,\dots, n_1$ and to $D$.
\end{itemize}

The idea is to set the parameters $\alpha$ and $\beta$ in such a way that $A$ becomes the unique Nash equilibrium. For any strategy profile $s$, there are $k$ players playing strategy $A_i$ and $n_1 - k$ players playing strategy $P_i$ in the set $G_1$, for some $0 \leq k \leq n_1$. By symmetry, it then suffices to look at profiles $S_k = (A_1,\dots,A_k,P_{k+1},\dots,P_{n_1},D,\dots,D)$ for $0 \leq k \leq n_1$. Furthermore, the first $k$ players playing $A_i$ all have the same cost, and also, the $n_1 - k$ players playing $P_i$ have the same cost. We can therefore focus on the costs of player $1$, denoted by $C^\rho_A(k)$, and that of player $n_1$, denoted by $C^\rho_P(k)$. We have
\begin{eqnarray}
C^\rho_A(k) &= &\beta (k-1) + \beta (1 + \rho(2-1))(n_1 - k) + 1 + \rho(n_2+k - 1) \nonumber \\
&=& (\beta - \beta(1+\rho) + \rho)k + (-\beta + \beta(1+\rho)n_1 + 1 + \rho(n_2 - 1)) \nonumber \\
&=& \rho(1 - \beta)\cdot k + (1 - \beta - \rho) + \beta(1+\rho)n_1 + \rho n_2 \nonumber
\end{eqnarray}
and
\begin{eqnarray}
C^\rho_P(k) &= & \alpha + \beta (n_1 - 1 - k)+ \beta(1 + \rho(2-1))k \nonumber \\
&=& \beta \rho \cdot k + \alpha + \beta(n_1 - 1)
\end{eqnarray}
\end{proof} 
We can set the parameters $\alpha$ and $\beta$ such that $C^\rho_A(k) = C^\rho_P(k-1)$, meaning that $S_k$ is a Nash equilibrium for every $k$ (we will create a unique Nash equilibrium in a moment), that is we take
$$
\rho(1-\beta) = \beta \rho \ \ \ \ \text{and} \ \ \ \ (1 - \beta - \rho) + \beta(1+\rho)n_1 + \rho n_2 = \alpha + \beta(n_1 - 1) - \beta \rho
$$
Note that the $-\beta \rho$ term on the far right of the second equation comes from the fact that we evaluate $C^\rho_P(\cdot)$ in $k-1$ (remember that $k$ denotes the number of players playing strategy $A_i$, so if a player would switch to $P_i$ this number decreases by $1$). Solving the left equation leads to $\beta = 1/2$. Inserting this in the right equation, and solving for $\alpha$, gives 
$$
\alpha = \rho \left( \frac{n_1}{2} + n_2 - \frac{1}{2} \right) + 1.
$$
We emphasize that $\alpha, \beta > 0$ for all $\rho \geq 0$. In order to make $A$ the unique Nash equilibrium, we can slightly increase $\alpha$ such that we get $C^\rho_A(k) < C^\rho_P(k-1)$
for all $k$ (which means that $A_i$ is a dominant strategy for player $i$). Note that this increase in $\alpha$ can be arbitrary small.
We have
$$
\frac{C^\sigma(A)}{C^\sigma(P)} = \frac{n_1\left[1 + \sigma(n_1+n_2 - 1) + \frac{1}{2}(n_1 - 1)\right] + n_2\left[1 + \sigma(n_1+n_2 - 1)\right]}{n_1\left[\rho(\frac{n_1 + 1}{2} + n_2 - 1) + 1 + \frac{1}{2}(n_1-1) \right] + n_2\left[1 + \sigma(n_2 - 1)\right]}
$$
Inserting $n_2 = a \cdot n_1$ for some rational $a > 0$, and sending $n_1 \rightarrow \infty$ gives a lower bound of
$$
f(a) = \frac{2\sigma(1+a)^2 + 1}{\rho(1+2a) + 1 + 2\sigma a^2}
$$
on the price of stability. Optimizing over $a > 0$ (this only works if $\rho < 2\sigma$) gives $a^* = -\frac{1}{2} + \sqrt{\frac{1}{4} + \frac{1}{2\sigma}}$ and $f(a^*)$ then yields the bound in the statement of the theorem. \qed

\section{Remarks and Additional Insights}\label{sec:misc}

We discuss the results of Sections \ref{sec:poa} and \ref{sec:pos} for the applications mentioned in Section \ref{sec:prelim} and obtain some additional application-specific insights and results. We conclude this section by showing that the price of stability results are not tight for the special case of symmetric network congestion games (as opposed to the price of anarchy results, which are tight for this class).

\begin{figure}[h!]
\centering
\scalebox{0.8}{
\begin{tikzpicture}

\draw[->] (0,0) -- (11,0) node[anchor=north] {$\rho$};
\draw	(0,0) node[anchor=north] {0}
		(3,0) node[anchor=north] {$h(1) \approx 0.625$}
		(5,0) node[anchor=north] {$1$}
		(8,0) node[anchor=north] {$2$};

\draw[->] (0,0) -- (0,6) node[anchor=east] {$\text{PoA}(\Gamma, \rho,1)$};

\draw[very thick] (8,4) -- (10,6);
\draw (10,5) edge[->] (9.5,5);   
\draw (10.5,5) node {$\rho + 1$};

\draw (3,2) edge[bend left=10, very thick] (8,4);   
\draw (6,3) edge[->] (6,3.3);   
\draw (6,2.6) node {$\dfrac{4\rho+1}{\rho+1}$};

\draw[dotted, very thick] (0,2) -- (3,2);
\draw (-0.5,2) node {$2.155$};

\draw (0.2,6) edge[bend right=25, very thick] (2.5,2);   
\draw (1.5,4.5) edge[->] (1,4);   
\draw (2.25,4.5) node {$\dfrac{4}{\rho(4 - \rho)}$};

\draw[dotted] (3,0) -- (3,2);
\draw[dotted] (8,0) -- (8,4);

\end{tikzpicture}}
\caption{Lower bounds on the price of anarchy for $\sigma = 1$. The bounds $(4\rho + 1)/(\rho + 1)$ and $\rho + 1$ are also tight upper bounds. The dotted horizontal line indicates the lower bound following from Theorem 3.7 \cite{Caragiannis2010Altruism}. The bound $4/(\rho(4 - \rho))$ is a lower bound for symmetric singleton congestion games proven in Theorem \ref{thm:pos_upper_network}. For general congestion games, a tight bound for $\text{PoA}(\Gamma, \rho,1)$ when $0 < \rho \leq h(1)$ is still an open problem.} 
\label{fig:sigma1}
\end{figure}
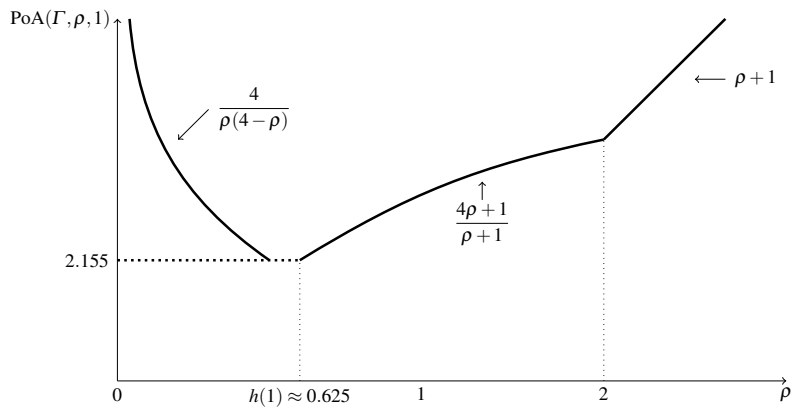

\smallskip
\textit{Altruism \cite{Caragiannis2010Altruism,Chen2014}.}  For $\sigma = 1$, we see that the resulting bound $(4\rho+1)/(\rho+1)$ is increasing in $\rho$. Especially in the context of altruism or taxes, this is not desirable. A natural question to ask is therefore if there exist collections $\mathcal{H}$ for which the price of anarchy is non-increasing as a function of $\rho$. 

The next theorem gives a sufficient condition for a class of instances to have the property of non-increasing price of anarchy, which can be interpreted as follows. For $\rho = 2$, it can be shown that the social optimum becomes a Nash equilibrium  under the player cost function $C^\rho_i(s)$, which implies that the price of stability is $1$. Nevertheless, it might happen that  worse Nash equilibria arise as well. The condition $\text{PoA}(\mathcal{H},2,1) = 1$ states that this does not happen, and even stronger, that all Nash equilibria of the game for $\rho = 2$ become social optima (it is said that the social optimum is \textit{strongly enforceable}).
\begin{theorem}\label{thm:poa_se}
Let $\mathcal{H}$ be a collection of congestion games. If $\text{PoA}(\mathcal{H},2,1) = 1$, then $\text{PoA}(\mathcal{H},\rho,1)$ is a non-increasing function for $1 \leq \rho \leq 2$.
\end{theorem}
In a technical report, Singh \cite{Singh} shows that the social optimum is strongly enforceable for symmetric network congestion games on series-parallel graphs. We can therefore conclude that the (altruistic) price of anarchy will be a non-increasing function of $\rho$. This is a remarkable result since, to the best of our knowledge, the classical price of anarchy is unknown (the best lower bound is given by Fotakis \cite{Fotakis2010}). 

\smallskip
\textit{Constant taxes \cite{Caragiannis2010}}. Caragiannis et al. \cite{Caragiannis2010} showed that the price of anarchy can be decreased to $2.155$ by the usage of universal tax functions (see also Section 1 and Figure \ref{fig:sigma1}), which improves significantly the classical bound of $2.5$. However, the price of stability increases from $1.577$ (for classical games) to $2.013$, for this specific set of tax functions.  Furthermore, from Theorem 3.7 \cite{Caragiannis2010} it follows that the price of anarchy can never be better than $2.155$ for $0 \leq \rho \leq h(1)$. In Theorem \ref{thm:pos_upper_network} we even show that the price of anarchy goes to infinity as $\rho \rightarrow 0$.

\smallskip
\textit{Risk sensitivity under uncertainty \cite{Piliouras2013}.} We do not only re-obtain the price of anarchy results for risk-neutral players and players applying Wald's minimax principle (worst-case players), but our results also give a tight bound for any convex combination (in terms of player costs) of risk-neutral and worst-case risk attitudes. Furthermore, we also obtain tight price of stability results for this model.

\smallskip
\textit{Approximate Nash equilibria \cite{Christodoulou2011}.} For $\sigma = 1$ and $1 \leq \rho \leq 2$, we obtain a bound of $(\sqrt{3} + 1)/(\sqrt{3} + \rho - 1)$ on the price of stability. In particular, this also yields the same bound on the price of stability for $\rho$-approximate equilibria. This bound was previously obtained by Christodoulou et al. \cite{Christodoulou2011}. Conceptually our approach is different since we obtain correctness of the bound through the observation that every Nash equilibrium in our framework yields an approximate equilibrium. In particular, this immediately yields a potential function  that can be used to carry out the technical details (namely the potential function that is exact for our congestion game). Nevertheless, the framework of Christodoulou et al. \cite{Christodoulou2011} is somewhat more general and might be used to obtain a tight bound for the price of stability of approximate equilibria (which is not known to the best of our knowledge).

\subsection{Price of stability for symmetric network congestion games}\label{sec:pos_network}
The price of anarchy bound of $(1+2\rho(1+\sigma))/(1 + \rho)$ obtained in Section \ref{sec:poa} is tight even for symmetric network congestion games with linear cost functions, as was shown in Theorem \ref{thm:poa_network_lower}. This is not true for the price of stability, which we will show here for the case $\sigma = 1$.

\begin{theorem}\label{thm:pos_upper_network} Let $\Gamma$ be a linear symmetric network congestion game, then 
$$
\text{PoS}(\Gamma, \rho, 1) \leq \left\{ \begin{array}{ll} 
4/(\rho(4-\rho)) & \ \ \ \text{ if } 0  \leq \rho  \leq 1 \\
4/(2+\rho) &  \ \ \ \text{ if } 1  \leq \rho \leq 2 \\ 
(2+\rho)/4 & \ \ \ \text{ if } 2  \leq \rho  < \infty
\end{array} \right.
$$
In particular, if $\Gamma$ is a symmetric congestion game on an extenstion-parallel\footnote{A graph $G$ is extension parallel if it consists of either (i) a single edge, (ii) a single edge and an extension-parallel graph composed in series, (iii) two extension-parallel graphs composed in parallel.} graph $G$, then the upper bounds even hold for the price of anarchy. All bounds are tight.
\end{theorem}
For $\rho \geq 1$, the bounds were previously shown by Caragiannis et al. \cite{Caragiannis2010Altruism} for the price of anarchy of singleton symmetric congestion games (which can be modeled on an extension-parallel graph).

Since any Nash equilibrium under the player cost $C_i^\rho(\cdot)$ is in particular a $\rho$-approximate Nash equilibrium, we also obtain the following result.
 
\begin{corollary}
The price of stability for $\rho$-approximate equilibria, with $1 \leq \rho \leq 2$, is upper bounded by $4/(2+\rho)$ for linear symmetric network congestion games.

\end{corollary}

\noindent \textbf{Acknowledgements.} We thank the anonymous referees for useful comments and one referee in particular for very detailed suggestions and pointing out important typos.

\newpage 
\bibliographystyle{splncs03}
\bibliography{references}

\begin{thebibliography}{10}
\providecommand{\url}[1]{\texttt{#1}}
\providecommand{\urlprefix}{URL }

\bibitem{Anshelevich2004}
Anshelevich, E., Dasgupta, A., Kleinberg, J., Tardos, E., Wexler, T.,
  Roughgarden, T.: The price of stability for network design with fair cost
  allocation. In: Proceedings of the 45th Annual IEEE Symposium on Foundations
  of Computer Science. pp. 295--304. FOCS '04, IEEE Computer Society,
  Washington, DC, USA (2004), \url{http://dx.doi.org/10.1109/FOCS.2004.68}

\bibitem{Caragiannis2012}
Caragiannis, I., Fanelli, A., Gravin, N., Skopalik, A.: Computing approximate
  pure nash equilibria in congestion games. SIGecom Exch.  11(1),  26--29 (Jun
  2012), \url{http://doi.acm.org/10.1145/2325713.2325718}

\bibitem{Caragiannis2010}
Caragiannis, I., Kaklamanis, C., Kanellopoulos, P.: Taxes for linear atomic
  congestion games. ACM Trans. Algorithms  7(1),  13:1--13:31 (Dec 2010),
  \url{http://doi.acm.org/10.1145/1868237.1868251}

\bibitem{Caragiannis2010Altruism}
Caragiannis, I., Kaklamanis, C., Kanellopoulos, P., Kyropoulou, M.,
  Papaioannou, E.: The Impact of Altruism on the Efficiency of Atomic
  Congestion Games, pp. 172--188. Springer Berlin Heidelberg, Berlin,
  Heidelberg (2010)

\bibitem{Chen2014}
Chen, P.A., de~Keijzer, B., Kempe, D., Sch\"{a}fer, G.: Altruism and its impact
  on the price of anarchy. ACM Trans. Econ. Comput.  2(4),  17:1--17:45 (Oct
  2014), \url{http://doi.acm.org/10.1145/2597893}

\bibitem{Christodoulou2005}
Christodoulou, G., Koutsoupias, E.: The price of anarchy of finite congestion
  games. In: Proceedings of the Thirty-seventh Annual ACM Symposium on Theory
  of Computing. pp. 67--73. STOC '05, ACM, New York, NY, USA (2005)

\bibitem{Christodoulou2011}
Christodoulou, G., Koutsoupias, E., Spirakis, P.G.: On the performance of
  approximate equilibria in congestion games. Algorithmica  61(1),  116--140
  (Sep 2011)

\bibitem{Correa2015}
Correa, J., de~Jong, J., de~Keijzer, B., Uetz, M.: The curse of sequentiality
  in routing games. In: Proceedings of the 11th International Conference on Web
  and Internet Economics - Volume 9470. pp. 258--271. WINE 2015,
  Springer-Verlag New York, Inc., New York, NY, USA (2015)

\bibitem{Fotakis2010}
Fotakis, D.: Congestion games with linearly independent paths: Convergence time
  and price of anarchy. Theor. Comp. Sys.  47(1),  113--136 (Jul 2010)

\bibitem{Koutsoupias1999}
Koutsoupias, E., Papadimitriou, C.: Worst-case equilibria. In: Proceedings of
  the 16th Annual Conference on Theoretical Aspects of Computer Science. pp.
  404--413. STACS'99, Springer-Verlag, Berlin, Heidelberg (1999)

\bibitem{MS96}
Monderer, D., Shapley, L.S.: Potential games. Games and Economic Behavior
  14(1),  124--143 (1996)

\bibitem{Piliouras2013}
Piliouras, G., Nikolova, E., Shamma, J.S.: Risk sensitivity of price of anarchy
  under uncertainty. In: Proceedings of the Fourteenth ACM Conference on
  Electronic Commerce. pp. 715--732. EC '13, ACM, New York, NY, USA (2013)

\bibitem{Rosenthal1973}
Rosenthal, R.W.: A class of games possessing pure-strategy {Nash} equilibria.
  International Journal of Game Theory  2,  65--67 (1973)

\bibitem{Singh}
Singh, C.: Marginal cost pricing in atomic congestion games, technical report

\end{thebibliography}

\newpage
\appendix

\section{Omitted Material of Section \ref{sec:prelim}}
The following lemma shows the equivalence between dynamic taxes and altruistic players.
\begin{lemma}
For $1 \leq \rho \leq 2$, a strategy profile $s  \in \times_i \mathcal{S}_i$ is a Nash equilibrium under the cost
$$
{C}_i^{\rho}(s) = \sum_{e \in s_i} c_e(x_e) + (\rho - 1) \sum_{e \in s_i} (x_e - 1)[c_e(x_e) - c_e(x_e-1)]
$$ if and only if it is a Nash equilibrium under the altruistic cost 
$$
\mathcal{A}_i^\rho(s) = (2 - \rho) C^0_i(s) + (\rho - 1) C^0(s)
$$
\label{lem:alt_nash}
\end{lemma}
\begin{proof}
Let $s$ be a Nash equilibrium under the perceived costs $\mathcal{A}_i^\gamma(s)$, and let $s^*$ be any other strategy profile. Furthermore, define $\gamma = \rho - 1$. Leaving out all the resources that are \textit{not} in the symmetric difference\footnote{Resources used by at most on strategy.} of $s_i$ and $s_i^*$ (for the social cost term), we find that the Nash condition is equivalent to
\begin{eqnarray}
(1 - \gamma)\sum_{e \in s_i \setminus s_i^*} c_e(x_e)  + \gamma \left( \sum_{e \in s_i \setminus s_i^*} x_ec_e(x_e) +  \sum_{e \in s_i^* \setminus s_i} x_ec_e(x_e)\right) & \leq &(1 - \gamma)\sum_{e \in s_i^* \setminus s_i} c_e(x_e + 1)   \nonumber \\
&  & + \ \gamma \sum_{e \in s_i \setminus s_i^*} (x_e - 1)c_e(x_e - 1)  \nonumber \\
& & + \ \gamma \sum_{e \in s_i^* \setminus s_i} (x_e + 1)c_e(x_e + 1) \nonumber
\end{eqnarray}
which is equivalent to
$$
\sum_{e \in s_i \setminus s_i^*} c_e(x_e)  + \gamma \sum_{e \in s_i \setminus s_i^*} (x_e - 1)[c_e(x_e) - c_e(x_e - 1)] \leq \sum_{e \in s_i^* \setminus s_i} c_e(x_e + 1) + \ \gamma \sum_{e \in s_i^* \setminus s_i} x_e[c_e(x_e + 1) - c_e(x_e)]
$$
which is equivalent to ${C}_i^{\gamma}(s) \leq {C}_i^{\gamma}(s_{-i},s_i^*)$ since the terms for $ e \in s_i \cap s_i^*$ do not change. \qed
\end{proof}

\newpage
\section{Omitted Material of Section \ref{sec:poa}}
\begin{rtheorem}{Lemma}{\ref{lem:main_ineq}}
Let $s$ be a Nash equilibrium under the cost functions ${C}_i^\rho(s)$ and let $s^*$ be a minimizer of $C^\sigma(\cdot)$. For $\rho, \sigma \geq 0$ fixed, if there exist $\alpha(\rho,\sigma), \beta(\rho,\sigma) \geq 0$ such that
$$
(1 + \rho\cdot x)y - \rho(x - 1)x - x \leq - \beta(\rho,\sigma)(1 + \sigma(x - 1))x + \alpha(\rho,\sigma)(1 + \sigma(y - 1))y
$$
for all non-negative integers $x$ and $y$, then
$$
\frac{C^\sigma(s)}{C^\sigma(s^*)} \leq \frac{\alpha(\rho,\sigma)}{\beta(\rho,\sigma)}.
$$
\end{rtheorem}
\begin{proof}
Without loss of generality, we may assume that $a_e = 1$ and $b_e = 0$. We then have
\begin{eqnarray}
\sum_i {C}_i^\rho(s) &=& \sum_e \rho (x_e - 1)x_e + \sum_e x_e \nonumber \\
 &=& \sum_e \rho [1 - \sigma + \sigma](x_e - 1)x_e +\rho x_e - \rho x_e + \sum_e x_e \nonumber \\
  &=& \rho \sum_e [1 + \sigma(x_e - 1)]x_e + \rho \sum_e (1 - \sigma)(x_e-1)x_e - x_e + \sum_e x_e \nonumber \\
    &=& \rho C^\sigma(s) + \rho \sum_e (1 - \sigma)(x_e-1)x_e + (1 - \rho) \sum_e x_e. \nonumber 
\end{eqnarray} Rewriting this, we find
\begin{eqnarray}
\rho \cdot C^\sigma(s) &= & \sum_i {C}_i^\rho(s) + \rho(\sigma - 1)\sum_e (x_e-1)x_e  + (\rho - 1) \sum_e x_e\nonumber \\
& \leq & \sum_i {C}_i^\rho(s^*_i,s_{-i}) + \rho(\sigma - 1)\sum_e (x_e-1)x_e  + (\rho - 1) \sum_e x_e\nonumber \\
& \leq & \sum_e [1 + \rho(x_e - 1 + 1)]x_e^* + \rho(\sigma - 1)\sum_e (x_e-1)x_e  + (\rho - 1) \sum_e x_e\nonumber \\
& = & \sum_e [1 + \rho x_e]x_e^* + \rho(\sigma - 1)(x_e - 1)x_e + (\rho - 1)x_e  \nonumber \\
& = & \sum_e [1 + \rho x_e]x_e^* + \rho[ 1 + \sigma(x_e - 1)]x_e - \rho (x_e - 1)x_e - x_e  \nonumber \\
& = & \sum_e [1 + \rho x_e]x_e^* - \rho (x_e - 1)x_e - x_e  +  \rho C^\sigma(s) \nonumber \\
& \leq & - \beta(\rho,\sigma)  C^\sigma(s) + \alpha(\rho,\sigma) C^\sigma(s^*) + \rho C^\sigma(s) \nonumber  
\end{eqnarray} 
Rearranging terms then gives the desired result.\qed 
\end{proof}

 The following proposition is used in the proof of Theorem \ref{lem:main_ineq} below.
\begin{proposition}\label{prop:f1}
For every $(x,y) \in \N^2 \setminus \{(1,0)\}$, we have
$$
f_1(x,y,\sigma) = 2y(y-1)\sigma^2 + [x^2 + 2y^2 - 2xy -x] \sigma + [x^2 - xy +2(y - x)] \geq 0
$$
for $\sigma \geq \sigma^* = 1/2 $.
\end{proposition}
\begin{proof}
Note that $2y(y-1) \geq 0$ for all $y \in \N$. Furthermore, 
$$
x^2 + 2y^2 - 2xy -x = (x - y)^2 + y^2 - x \geq (x-y)^2 + (y-x) \geq 0
$$
for all $(x,y) \in \N^2$, using the fact that $a^2 - a \geq 0$ for all $a \in \N$. This means that $f_1(x,y,\sigma)$ is non-decreasing, and hence it suffices to prove the statement for $\sigma^* = 1/2$. We need to prove
$$
\frac{1}{2}y(y-1) + \frac{1}{2}[x^2 + 2y^2 - 2xy -x] + [x^2 - xy +2(y - x)] \geq 0,
$$
or equivalent,
$$
y(y-1) + x^2 + 2y^2 - 2xy - x + 2x^2 - 2xy + 4(y-x) \geq 0.
$$
Simplifying gives
$$
3x^2 + 3y^2 - 4xy + 3y - 5x \geq 0
$$
which is equivalent to 
$$
2(x -y)^2 + x(x-5) + y(y+3) \geq 0
$$
and this last formulation is clearly true for all pair $(x,y)$ with $x \geq 5$. For $x = 4$, we find $2(4 - y)^2 - 4 + y(y+3) \geq 0$ which is clearly true for $y \geq 1$, and for $y = 0$ in can be checked through inspection. For $x = 3$, we find $2(3-y)^2 - 6 + y(y+3) \geq 0$ which is clearly true for $y \geq 2$. For $y \in \{0,1\}$, it can be check through inspection. For $x = 2$, we find $2(2 - y)^2 - 6 + y(y+3) \geq 0$, which is again clear for $y \geq 2$, and for $y \in \{0,1\}$ it can be check through inspection. For $x = 1$, we find $2(1 - y)^2 - 4 + y(y+3) \geq 0$, which is clearly true for $y \geq 1$. For $y = 0$ the inequality does not hold, but this is the case $(x,y) = (1,0)$ that we do not consider. For $x = 0$, it is clearly true. \qed
\end{proof} \medskip

\begin{rtheorem}{Theorem}{\ref{thm:poa_upper}}
Let $s$ be a Nash equilibrium under the cost functions ${C}_i^\rho(s)$ and let $s^*$ be a minimizer of $C^\sigma(\cdot)$. Then
\begin{equation}\label{eq:main_poa_result_app}
\frac{C^\sigma(s)}{C^\sigma(s^*)} \leq \frac{2\rho(1 + \sigma) + 1}{\rho + 1} 
\end{equation}
\begin{enumerate}[i)]
\item if $\frac{1}{2} \leq \sigma \leq \rho \leq 2\sigma$,
\item if $\sigma = 1$ and $h(\sigma) \leq \rho \leq 2 \sigma$,
\end{enumerate}
where $h(\sigma) = g(1 + \sigma + \sqrt{\sigma(\sigma + 2)}, \sigma)$ is the optimum of the function
$$
g(a,\sigma) = \frac{\sigma(a^2 - 1)}{(1 + \sigma)a^2 - (2\sigma + 1)a + 2\sigma(\sigma+1)}.
$$
Furthermore, there exists a function $\Delta = \Delta(\sigma)$ which the property that, for any fixed $\sigma_0 \geq 1/2$: if $\Delta(\sigma_0) \geq 0$, then (\ref{eq:main_poa_result_app}) is true for all $h(\sigma_0) \leq \rho \leq 2\sigma_0$ (the function $\Delta$ can be found in the proof below).
\end{rtheorem}
\begin{proof}
We show the following inequality,
\begin{equation}\label{eq:main_ineq_text}
(1 + \rho\cdot x)y - \rho(x - 1)x - x \leq - \frac{1 + \rho}{1 + 2\sigma}(1 + \sigma(x - 1))x + \frac{2\rho(1 + \sigma) + 1}{1 + 2\sigma}(1 + \sigma(y - 1))y.
\end{equation}
Multiplying with $(1 + 2\sigma)$ we obtain the equivalent formulation
$$
(1 + 2\sigma)\left[(1 + \rho\cdot x)y - \rho(x - 1)x - x\right] \leq -(1 + \rho)(1 + \sigma(x - 1))x + (2\rho(1 + \sigma) + 1)(1 + \sigma(y - 1))y
$$
which we can rewrite to $f_1(x,y,\sigma) \rho + f_2(x,y,\sigma) \geq 0$ where
\begin{eqnarray}
f_1(x,y,\sigma) &=& -(1 + \sigma(x-1))x + 2(1 + \sigma)(1 + \sigma(y-1))y + (1 + 2\sigma)((x-1)x - xy) \nonumber \\
&=& 2y(y-1)\sigma^2 + (-(x-1)x + 2(y-1)y + 2y + 2x(x-1) - 2xy)\sigma \nonumber \\
&+& ( -x + 2y + (x-1)x - xy) \nonumber \\
& = & 2y(y-1)\sigma^2 + [x^2 + 2y^2 - 2xy -x] \sigma + [x^2 - xy +2(y - x)] \nonumber
\end{eqnarray}
and
\begin{eqnarray}
f_2(x,y,\sigma)& =& -(1 + \sigma(x-1))x + (1 + \sigma(y-1))y + (1+2\sigma)(x - y) \nonumber \\
 & = & \sigma y(y-1) - \sigma x(x-1) + 2\sigma (x-y) \nonumber \\
  & = & \big(y^2 - x^2 + 3(x-y)\big)\sigma  \nonumber 
\end{eqnarray}
We first consider the case $(x,y) = (1,0)$, since then we do not have $f_1(x,y,\sigma) \geq 0$. Substituting the values for $x$ and $y$, we obtain
$$
-\rho + 2\sigma \geq 0
$$
which is true if and only if $\rho \leq 2\sigma$. 

\textbf{Case i).}
For the pair $(x,y) = (1,0)$, the inequality is true if and only if $\rho \leq 2\sigma$. For all other pairs, we have $f_1(x,y,\sigma) \geq 0$, and hence 
$$
f_1(x,y,\sigma) \rho + f_2(x,y,\sigma) \geq f_1(x,y,\sigma) \sigma + f_2(x,y,\sigma) 
$$
meaning that is suffices to show that $f_1(x,y,\sigma) \sigma + f_2(x,y,\sigma) \geq 0$.
After dividing by $\sigma$, we see that this is equivalent to
$$
2y(y-1)\sigma^2 + [x^2 + 2y^2 - 2xy -x] \sigma + [x^2 - xy +2(y - x)] + \big(y^2 - x^2 + 3(x-y)\big) \geq 0
$$
which is equivalent to
$$
2y(y-1)\sigma^2 + [x^2 + 2y^2 - 2xy -x] \sigma + [y^2 - xy + (x - y)] \geq 0
$$
Again, we see that the terms before $\sigma^2$ and $\sigma$ are non-negative for all $x,y \in \N$ (see proof of Proposition \ref{prop:f1}), meaning that if the inequality holds for some $\sigma^*$, then it holds for all $\sigma \geq \sigma^*$. We take $\sigma^* = 1/2$. Multiplying the resulting inequality with $2$, we find
$$
y(y-1)+ [x^2 + 2y^2 - 2xy -x] + 2[y^2 - xy + (x - y)] \geq 0
$$
which is equivalent to
$$
x^2 + 5y^2 - 4xy -3y +x \geq 0.
$$
This can be rewritten as
$$
(x - 2y)^2 + y(y-3) + x \geq 0
$$
which is clearly true for all $y \notin \{1,2\}$. For $y = 1$, we find $(x - 2)^2  - 2 + x \geq 0$. This is clearly true for all $x \geq 2$. For $x \in \{0,1\}$, it can be checked through inspection. For $y = 2$, we find $(x - 4)^2 - 2 + x \geq 0$. This is again clearly true for $x \geq 2$, and can be check through inspection for $x \in \{0,1\}$. 

\textbf{Case ii).} Now let $(x,y) \in \N^2 \setminus \{(1,0)\}$, then  $f_1(x,y,\sigma) \geq 0$ by Proposition \ref{prop:f1}, meaning that $f_1(x,y,\sigma) \rho + f_2(x,y,\sigma)$ is non-decreasing in $\rho$. From the proof of Proposition \ref{prop:f1}, it follows that $f_1(x,y,\sigma) = 0$ if and only if $(x,y) \in \{(1,1), (2,1)\}$ (which can be seen by checking all the cases). Note that this observation is independent of $\sigma$. For $(x,y) \in \{(1,1), (2,1)\}$ it also holds that $f_2(x,y,\sigma) = 0$, which implies that $f_1(x,y,\sigma) \rho + f_2(x,y,\sigma) = 0$ for every $\rho$. Therefore, we can focus on pairs $(x,y)$ for which $f_1(x,y,\sigma) > 0$. 
It follows that any $\rho^*$ for which 
$$
\rho^* \geq \sup_{x,y \in \N: f_1(x,y,\sigma) > 0} - \frac{f_2(x,y,\sigma)}{f_1(x,y,\sigma)}.
$$
yields the inequality for all $\rho \geq \rho^*$. It is not hard to see that this supremum is indeed finite, for every fixed $\sigma$. It can be proved that $f_1(x,y,\sigma) \rho' + f_2(x,y,\sigma) \geq 0$ holds for some large constant $\rho'$, which then serves as an upper bound on the supremum. For the pair $(x,y) = (0,1)$, we find $-f_2/f_1 = \sigma/(1+\sigma)$, but we will see later that the supremum on the other pairs obtained is larger than $\sigma/(1+\sigma)$.

Note that by now, we can focus on pairs in $\{(x,y) : x \geq 1, y \geq 2\}$, since for all other pairs we have either proven the inequality or given $-f_2/f_1$, that is, we are interested in 
\begin{equation} \label{eq:sup}
\sup_{\{(x,y) : x \geq 1, y \geq 2\}} - \frac{f_2(x,y,\sigma)}{f_1(x,y,\sigma)}.
\end{equation}
Note that $f_2(x,y,\sigma) = \big(y^2 - x^2 + 3(x-y)\big)\sigma = (x+y - 3)(y - x) \geq 0$ if $y \geq x$ (using that $x+y \geq 3$ for $(x,y) \in \{(x,y) : x \geq 1, y \geq 2\}$). Hence, if $y \geq x$, we have $-f_2/f_1 \leq 0$, so those pairs are not relevant for the supremum (if it follows that the upper bound on the supremum for all other pairs is positive, which we will indeed see later). Therefore, we can focus on pairs with $y < x$.

We substitute $x = ay$ for some (rational) $a > 1$. Note that
\begin{equation}\label{eq:upperbound}
\sup_{a  \in \R_{> 1}} \sup_{ y \geq 2} - \frac{f_2(ay,y,\sigma)}{f_1(ay,y,\sigma)}
\end{equation}
provides an upper bound on (\ref{eq:sup}). We have
$$
f_1(ay,y,\sigma) = [(1+\sigma)a^2 - (2\sigma + 1)a + 2\sigma(\sigma+1)]y^2 - [(2+\sigma)a + 2\sigma^2 - 2]y
$$
and
$$
- f_2(ay,y,\sigma) = [(a^2 - 1)\sigma]y^2 + [3(1-a)\sigma]y
$$
We determine an upper bound on the expression
\begin{eqnarray}
- \frac{f_2(ay,y,\sigma)}{f_1(ay,y,\sigma)} &=& \frac{[(a^2 - 1)\sigma]y^2 + [3(1-a)\sigma]y}{[(1+\sigma)a^2 - (2\sigma + 1)a + 2\sigma(\sigma+1)]y^2 - [(2+\sigma)a + 2\sigma^2 - 2]y} \nonumber \\
 & = &\frac{[(a^2 - 1)\sigma]y + [3(1-a)\sigma]}{[(1+\sigma)a^2 - (2\sigma + 1)a + 2\sigma(\sigma+1)]y - [(2+\sigma)a + 2\sigma^2 - 2]} \nonumber  \\
 &=& \frac{\alpha y + \beta}{\gamma y - \delta} \label{eq:alpha}
\end{eqnarray}
for $y \geq 2$. Elementary calculus shows that the derivative with respect to $y$ of (\ref{eq:alpha}) is given by $-(\alpha \delta + \gamma \beta)/(\gamma y - \delta)^2$, which means the expression in (\ref{eq:alpha}) is  non-decreasing or non-increasing in $y$. 
We have
\begin{eqnarray}
\alpha \delta + \gamma \beta & =& (a^2 - 1)\sigma[(2+\sigma)a + 2\sigma^2 - 2] + 3(1-a)[(1+\sigma)a^2 - (2\sigma+1)a +2\sigma(1+\sigma)] \nonumber \\
&=& (1-a)\sigma \left[-(1+a)((2+\sigma)a + 2\sigma^2 - 2) + 3((1+\sigma)a^2 - (2\sigma+1)a +2\sigma(1+\sigma))\right] \nonumber \\
&=& (1-a)\sigma \left[(3(1+\sigma) - (2+\sigma))a^2 + (2 - (2+\sigma) - 2\sigma^2 - 3(2\sigma + 1))a + 2 - 2\sigma^2 + 6\sigma(1 + \sigma)\right] \nonumber \\
&=& (1-a)\sigma \left[(2\sigma + 1)a^2 - (2\sigma^2 + 7\sigma + 3)a + (4\sigma^2 + 6\sigma + 2)\right] \nonumber \\
&=& (1-a)\sigma \left[(2\sigma + 1)a^2 - (2\sigma + 1)(\sigma+3)a + (2\sigma + 1)(2\sigma + 2)\right] \nonumber \\
&=& (1-a)\sigma(1+2\sigma) \left[a^2 - (\sigma+3)a + (2\sigma + 2)\right] \nonumber \\
&=& (1-a)\sigma(1+2\sigma) \left[\left(a - \frac{\sigma+3}{2}\right)^2 - \frac{1}{4} (1- \sigma)^2\right] \label{eq:derivative}
\end{eqnarray}

\textit{Intermezzo. If we consider the function $x_2 = (\alpha x_1 + \beta)/(\gamma x_1 - \delta)$, we see it has vertical asymptote at $x_1^* = \delta/\gamma$. 
We claim that $x_1^* < 2$. Note that, since $a > 1$, we have $\delta > 0$ for all $\sigma \geq 0$. If $\gamma < 0$ then $x_1^* < 0$. If $\gamma > 0$, we claim that $x_1^* < 2$. This is equivalent to showing that
$$
(2+\sigma)a + 2\sigma^2 - 2 < 2(1+\sigma)a^2 - 2(2\sigma +1)a + 4\sigma(\sigma+1),
$$
which holds if and only if
$$
2(1+\sigma)a^2 - (5\sigma + 4)a + 2(1+\sigma)^2 = 2(1+\sigma)\left( \left[a - \frac{5\sigma + 4}{4(1+\sigma)} \right]^2 - \frac{1}{4}\left[\frac{5\sigma + 4}{2(1+\sigma)} \right]^2 + (1+\sigma)\right) > 0.
$$
If now suffices to show that $- \frac{1}{4}\left[\frac{5\sigma + 4}{2(1+\sigma)} \right]^2 + (1+\sigma) > 0$, but this is true for all $\sigma > 0$, hence, the claim is proven.} \\

\textbf{The situation $\sigma = 1$.} It follows that the expression in (\ref{eq:derivative}) is non-positive for all $a > 1$, which implies that $-(\alpha \delta + \gamma \beta)/(\gamma y - \delta)^2 \geq 0$ and hence $-f_2/f_1$ is non-decreasing in $y \geq 2$ for every $a > 1$ (using the intermezzo). We then have
$$
\lim_{y \rightarrow \infty} -\frac{f_2(ay,y,\sigma)}{f_1(ay,y,\sigma)} =  \frac{\sigma(a^2 - 1)}{(1 + \sigma)a^2 - (2\sigma + 1)a + 2\sigma(\sigma+1)} =:h_1(a,\sigma)
$$
and maximizing this function over $a \in \R_{>1}$, we find the optimum
\begin{equation}\label{eq:a*}
a^*(\sigma) = 1 + \sigma + \sqrt{\sigma(\sigma + 2)}.
\end{equation}

\textbf{The situation $1/2 \leq \sigma < 1$.}
More generally, for any $\sigma < 1$ it holds that $\alpha \delta + \gamma \beta \leq 0$ if and only if $a \notin (1+\sigma,2)$. In particular for every $a \notin (1+\sigma,2)$, we can then show that 
\begin{equation}\label{eq:general_sigma}
\sup_{ y \geq 2} - \frac{f_2(ay,y,\sigma)}{f_1(ay,y,\sigma)} \leq \lim_{y \rightarrow \infty} -\frac{f_2(a^*y,y,\sigma)}{f_1(a^*y,y,\sigma)}
\end{equation}
with $a^*$ as in (\ref{eq:a*}) using the same argument as in the case $\sigma = 1$. The intermezzo implies that if the expression (\ref{eq:alpha}) is non-increasing in $y$, which is the case when $a \in (1 + \sigma,2)$, then the maximum value is attained in $y = 2$. That is, we are interested in the expression $-f_2(2a,2,\sigma)/f_1(2a,2,\sigma)$, and in particular, we want to show that the supremum over $a \in (1+\sigma, 2)$ does not exceed the right hand side of (\ref{eq:general_sigma}), i.e., the supremum over all $a \notin (1+\sigma,2)$.

 Because of the discussion in the above, it suffices to study 
\begin{eqnarray}
-\frac{f_2(2a,2,\sigma)}{f_1(2a,2,\sigma)} &=& \frac{[(a^2 - 1)\sigma]2 + [3(1-a)\sigma]}{[(1+\sigma)a^2 - (2\sigma + 1)a + 2\sigma(\sigma+1)]2 - [(2+\sigma)a + 2\sigma^2 - 2]}  \nonumber \\
&=&  \frac{\sigma(2a^2-3a+1)}{2(1+\sigma)a^2 - (5\sigma + 4)a + 2(1+\sigma)^2} =:h_2(a,\sigma)
\end{eqnarray}
for $a \in (1+\sigma,2)$. This expression, for $a > 1$, is maximized for $b^*(\sigma) = 1 + \sigma + \sqrt{\sigma(\sigma+1/2)}$, which in particular gives an upper bound for $a \in (1+\sigma,2)$). 

It now suffices to show that 
$$
\Delta(\sigma) := h_1(a^*(\sigma),\sigma) - h_2(b^*(\sigma),\sigma) \geq 0,
$$
since this implies that the supremum over $a > 1$ in (\ref{eq:upperbound}) is attained at some $a \notin (1+\sigma,2)$. We have checked this numerically (see Figure \ref{fig:sigma}).

\textbf{The situation  $\sigma > 1$. } We can use similar reasoning as in the previous case, but now the expression in (\ref{eq:alpha}) is non-increasing for $a \in (2,1+\sigma)$. Note that this does not affect the reasoning in the previous case, since we maximize over all $a > 1$ when obtaining $b^*(\sigma)$. \qed
%
\begin{figure}[h!]
\centering
\includegraphics[scale=0.9]{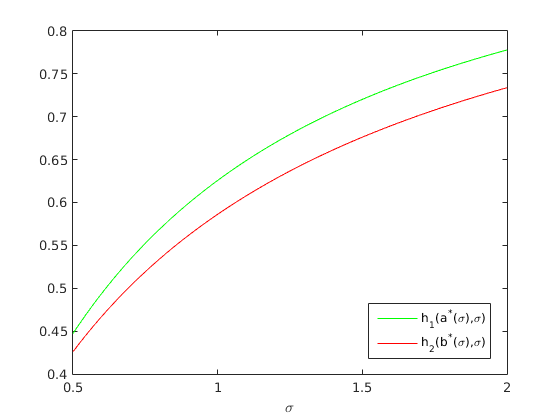}
\caption{Numerical verification that $h_2(b^*(\sigma),\sigma) \leq h_1(a^*(\sigma),\sigma)$ for $1/2 \leq \sigma \leq 2$ with step size $\Delta = 10^{-6}$.}
\label{fig:sigma}
\end{figure} 
\end{proof}

\noindent The following theorem shows that the bound in Theorem \ref{thm:poa_upper} is also a lower bound for general congestion games and arbitrary $\rho, \sigma \geq 0$. We generalize the construction of Christodoulou and Koutsoupias \cite{Christodoulou2005}, who showed the lower bound for the classical case $\rho = \sigma = 1$.  This construction is also used in the risk-uncertainty model of Nikolova et al. \cite{Piliouras2013}, and the altruism model of Chen et al. \cite{Chen2014}.
\begin{theorem}
For $\rho,\ \sigma \in \R_{>0}$ fixed, and players with cost functions ${C}_i^\rho(s)$, there exist linear congestion games such 
$$
\frac{C^\sigma(s)}{C^\sigma(s^*)} \geq \frac{2\rho(1 + \sigma) + 1}{\rho + 1}
$$
where $s$ is a Nash equilibrium, and $s^*$ a socially optimal strategy profile.
\label{thm:poa_lower}
\end{theorem}
\begin{proof}
We construct a congestion game of $n \geq 3$ players and $|E| = 2n$ resources with price of anarchy greater than or equal to $(2\rho(1 + \sigma) + 1)/(\rho + 1)$. The set $E$ is divided in the sets $E_1 = \{h_1,\dots,h_n\}$ and $E_2 = \{g_1,\dots,g_n\}$. Player $i$ has two pure strategies: $\{h_i,g_i\}$ and $\{h_{i+1},g_{i-1},g_{i+1}\}$, where the indices appear as $i$ mod $n$. The latency functions of the elements in $E_1$ are $c_e(x) = x$, whereas the latency functions of the elements in $E_2$ are $c_e(x) = \rho x$.\\
\indent  Regardless which strategy player $i$ plays, he always uses at least one resource from both $E_1$ and $E_2$, implying that $C^\sigma_i(s) \geq \rho + 1$. This implies that
\begin{equation}\label{eq:general_t}
C^\sigma(t)  = \sum_{i \in N} C^\sigma_i(s) \geq (\rho + 1)n
\end{equation}
for every strategy profile $t$, and in particular for a social optimum $s^*$. 

We will now show that the strategy profile $s$ where every agent $i$ plays its second strategy $\{h_{i+1},g_{i-1},g_{i+1}\}$ is a Nash equilibrium. We have
$
{C}_i^\rho(s) = 2 \rho [1 + \rho(2 - 1)] + 1 = 2\rho^2 + 2\rho + 1.
$
If some agent $i$ deviates to its first strategy $s_i'$, we have
$
{C}_i^\rho(s_i',s_{-i}) = \rho[1 + \rho(3 - 1)] + (1 + \rho(2-1)) = 2\rho^2 + 2\rho + 1,
$
since there are then three agents using $g_i$ and two agents using $h_i$. This shows that $s$ is a Nash equilibrium. The social cost of this strategy $s$ is
\begin{equation}\label{eq:general_s}
C^\sigma(s) = n(1 + 2\rho [1 + \sigma(2-1)]) = \left(1 + 2\rho(1+\sigma) \right) n.
\end{equation}
Combining (\ref{eq:general_s}) with (\ref{eq:general_t}) then gives the desired result.\qed \medskip
\end{proof}

\begin{rtheorem}{Theorem}{\ref{thm:above_twosigma}}
Let $s$ be a Nash equilibrium under the cost functions ${C}_i^\rho(s)$ and let $s^*$ be a minimizer of $C^\sigma(\cdot)$. Then $C^\sigma(s)/C^\sigma(s^*) \leq \rho + 1$ for $1 \leq 2\sigma \leq \rho$. Furthermore, this bound is tight.
\end{rtheorem}
\begin{proof}
We show that
$$
(1 + \rho\cdot x)y - \rho(x - 1)x - x \leq - (1 + \sigma(x - 1))x + (1+\rho)(1 + \sigma(y - 1))y
$$
and then the result follows from Lemma \ref{lem:main_ineq}, with $h = 1$ and $g = 1+\rho$. Rewriting gives the equivalent statement
\begin{equation}\label{eq:twosigma}
\left[ y + \sigma y(y-1) - xy + x(x-1)\right]\rho + \sigma \left[y(y-1) - x(x-1)\right] \geq 0.
\end{equation}
We first show that $\left[ y + \sigma y(y-1) - xy + x(x-1)\right] \geq 0$ for all $\sigma \geq 1/2$. If suffices to show this claim for $\sigma = 1/2$, since $y(y-1) \geq 0$ for all $y \in \N$. We have
$$
y + \frac{1}{2} y(y-1) - xy + x(x-1) = \frac{1}{2}\left[ \left(x - y - \frac{1}{2}\right)^2 - \frac{1}{4} + x(x-1)\right]
$$
and this last expression is clearly non-negative for all $x,y \in \N$ (since the quadratic term is always at least $1/4$). 

It now suffices to show (\ref{eq:twosigma}) for $\rho = 2\sigma$, since we have shown that the expression is a non-decreasing affine function of $\rho$, for every fixed $\sigma \geq 1/2$. Substituting $\rho = 2\sigma$ and dividing (\ref{eq:twosigma}) by $\sigma$, we get the equivalent statement
\begin{equation}\label{eq:twosigma_2}
2 \left[ y + \sigma y(y-1) - xy + x(x-1)\right] + \left[y(y-1) - x(x-1)\right] \geq 0
\end{equation}
which we will show to be non-negative for all non-negative integers $x$ and $y$ and $\sigma \geq 1/2$. Again, it suffices to show the statement for $\sigma = 1/2$. The statement in (\ref{eq:twosigma_2}) is then equivalent to
$$
\left(x - y - \frac{1}{2}\right)^2 - \frac{1}{4} + y\left(y-1\right)
$$
which is clearly non-negative for all $x, y \in \N$. 

The tightness can be obtained by considering the following game on four resources with two players. Player $A$ has strategies $\{\{1\},\{2,4\}\}$ and player $B$ has strategies $\{\{2\},\{1,3\}\}$. Resources $e = 1,2$ have cost function $c_e(x) = x$ and resources $e = 3,4$ have cost function $c_e(x) = \rho x$. The optimum $s^* = (\{1\},\{2\})$ has cost $2$, whereas the Nash equilibrium $s = (\{2,4\},\{1,3\}$ has cost $2(1+\rho)$. \qed
\end{proof} \medskip

\newpage
\section{Omitted Material of Section \ref{sec:pos}}
\noindent \textbf{Lemma \ref{lem:main_ineq_pos}.}
For all non-negative integers $x$ and $y$, and $\sigma \geq 0$, we have
$$
\left(x - y + \frac{1}{2}\right)^2- \frac{1}{4} + 2\sigma x(x - 1) + (\sqrt{\sigma(\sigma+2)} + \sigma)[y(y-1) - x(x-1)] \geq 0.
$$
\begin{proof}
The inequality is clearly true for all $y \geq x$, so we assume that $y < x$.
Rewriting the expression gives
$$
(1 + \sigma + \sqrt{\sigma(\sigma+2)})y^2 - 2xy + (1 + \sigma - \sqrt{\sigma(\sigma+2)})x^2 - (1 + \sigma + \sqrt{\sigma(\sigma+2)})y + (1 - \sigma + \sqrt{\sigma(\sigma+2)})x \geq 0.
$$
Multiplying with $1 + \sigma - \sqrt{\sigma(\sigma+2)}$, which is non-negative for all $\sigma \geq 0$, gives
$$
y^2 - 2\left(1 + \sigma - \sqrt{\sigma(\sigma+2)}\right)xy + \left(1 + \sigma - \sqrt{\sigma(\sigma+2)}\right)^2x^2 - y
$$
$$
+  \ (1 + \sigma - \sqrt{\sigma(\sigma+2)})(1 - \sigma + \sqrt{\sigma(\sigma+2)})x \ \geq \ 0,
$$
using the fact that $(1+\sigma + \sqrt{\sigma(\sigma+2)})(1 + \sigma - \sqrt{\sigma(\sigma+2)}) = 1$. This is equivalent to
$$
\left(\left(1+\sigma - \sqrt{\sigma(\sigma+2)}\right)x - y + \frac{1}{2} \right)^2 + \left(1+\sigma - \sqrt{\sigma(\sigma+2)}\right)\left(\right[1+\sigma - \sqrt{\sigma(\sigma+2)}\left] - 1\right)x - \frac{1}{4} \geq 0.
$$
We now substitute $c = 1 + \sigma - \sqrt{\sigma(\sigma+2)} \geq 0$ in order to obtain the equivalent formulation
\begin{equation}\label{eq:c}
\left(c \cdot x - y + \frac{1}{2} \right)^2 + c(1-c)x - \frac{1}{4} \geq 0
\end{equation}
for $0 \leq c < 1$, since the function $c(\sigma) = 1 + \sigma - \sqrt{\sigma(\sigma+2)}$ is  bijective from $\R$ to $[0,1)$. For $x = 0$, the inequality reduces to $(\frac{1}{2} - y)^2 - 1/4 \geq 0$ which is true for all $y \in \N$. For $x = 1$, we get the equivalent formulation
$$
(y-1)\left(y - 2c\right) \geq 0,
$$
which is clearly true for $y = 1$. For $y = 0$, it follows from the fact that $c \geq 0$. For $y \geq 2$ if follows from the fact that $y - 2c \geq 0$ for all $y \geq 2$, since $0 \leq c < 1$. This completes the case $x = 1$. 

For $x \geq 2$, we rewrite the expression $(\ref{eq:c})$ to
\begin{equation}\label{eq:c_rewritten}
x(x - 1)c^2 + 2x(1 - y)c + y(y-1) \geq 0
\end{equation}
If $y = 0$, the expression in (\ref{eq:c_rewritten}) is clearly non-negative for all $x \geq 2$ and $0 \leq c < 1$. For $y \geq 1$, note that $g(c) = x(x - 1)c^2 + 2x(1 - y)c + y(y-1)$ is a quadratic and convex function for all fixed $x$ and $y$. Therefore, in particular, for any $x$ and $y$ fixed, it suffices to show that the inequality holds for the minimizer of $g$, which is $c^* = (y-1)/(x-1)$ (which can be found by differentiating with respect to $c$). Note that $0 \leq c^* < 1$ by our assumption that $y \geq 1$ and $y < x$ (made in the beginning of the proof). Substituting implies that it suffices to show that
$$
\frac{x(x-1)(y-1)^2}{(x-1)^2} + \frac{2x(1-y)(y-1)}{x-1} + y(y-1) \geq 0.
$$
Multiplying the expression with $(x-1)$ implies that it now suffices to show that
$$
x(y-1)^2 - 2x(y-1)^2 + y(y-1)(x-1)  \geq 0
$$
for all $1 \leq y < x$. This is always true since
\begin{eqnarray}
x(y-1)^2 - 2x(y-1)^2 + y(y-1)(x-1) & = & -x(y-1)^2 + y(y-1)(x-1) \nonumber \\
& = & (y-1) \left[-x(y-1) + y(x-1)\right] \nonumber \\
& = & (y-1) (x - y) \nonumber \\
& \geq & 0 \nonumber
\end{eqnarray}
whenever $1 \leq y < x$. This completes the proof. \qed \medskip
\end{proof}

\medskip
The bound of $2$ on the price of stability for generalized affine congestion games requires some additional arguments: 
The bound in  Theorem \ref{thm:pos_upper} with $\rho = \sigma$ is only valid for $\sigma \geq 1/4$ (because otherwise the lower bound on $\rho$ is not satisfied). Nevertheless for $0 \leq \sigma \leq 1/4$, the corresponding cost functions $c_e(x) = \sigma x + (1-\sigma)$ have non-negative constants and thus the price of stability for classical congestion games applies here. That is, we have
$$
\text{PoS}(\mathcal{A}') = \max\left\{1.577, \  \sup_{\sigma \geq 1/4} \left\{1 + \sqrt{\sigma/(\sigma+2)}\right\}\right\} = 2.
$$

\newpage 
\section{Omitted Material of Section \ref{sec:misc}}

\begin{rtheorem}{Theorem}{\ref{thm:poa_se}}
Let $\mathcal{H}$ be a collection of congestion games. If $\text{PoA}(\mathcal{H},1,2) = 1$, then $\text{PoA}(\mathcal{H},1,\rho)$ is non-increasing function for $1 \leq \rho \leq 2$.
\end{rtheorem}
\begin{proof}
Let $\alpha = \rho - 1$. Suppose that $\text{PoA}(\mathcal{H},1,\alpha) =: \text{PoA}(\alpha)$ is not non-increasing, then there exist $x < y \in [0,1]$ such that $\text{PoA}(y) > \text{PoA}(x)$. We also know, by assumption, that $\text{PoA}(x) \geq \text{PoA}(1) = 1$, since the price of anarchy is always lower bounded by $1$ (note that this also implies that $y \neq 1$). This means that
$$
\max \{\text{PoA}(x), \text{PoA}(1)\} < \text{PoA}(y).
$$ 
However, if we write $y = \gamma\cdot x + (1- \gamma)\cdot 1$ for some $\gamma \in [0,1]$, then Theorem 10.2 \cite{Chen2014} implies that 
$$
\text{PoA}(y) \leq \max \{\text{PoA}(x),\text{PoA}(1)\}
$$
which is a contradiction.\qed
\end{proof}

\section{Omitted Material of Section \ref{sec:pos_network}}

\begin{rtheorem}{Theorem}{\ref{thm:pos_upper_network}}
Let $\Gamma$ be a symmetric network congestion game with linear cost functions, then 
$$
\text{PoS}(\Gamma, \rho, 1) \leq \left\{ \begin{array}{ll} 
4/(\rho(4-\rho)) & \ \ \ \text{ if } 0  \leq \rho  \leq 1 \\
4/(2+\rho) &  \ \ \ \text{ if } 1  \leq \rho \leq 2 \\ 
(2+\rho)/4 & \ \ \ \text{ if } 2  \leq \rho  < \infty
\end{array} \right.
$$
In particular, if $\Gamma$ is a symmetric congestion game on an extenstion-parallel \footnote{A graph $G$ is extension parallel if it consists of either (i) a single edge, (ii) a single edge and an extension-parallel graph composed in series, (iii) two extension-parallel graphs composed in parallel.} graph $G$, then the upper bounds even hold for the price of anarchy. All bounds are tight.
\end{rtheorem} \medskip

\noindent The remainder of this section is dedicated to the proof of Theorem \ref{thm:pos_upper_network}. 
We will refer to strategy profiles as flows, since we can interpret symmetric network congestion games as a flow problem in which players each have to route one unit of unsplittable flow from a given source to a given target. To be precise, for a graph $G = (V,E)$ and given $s,t \in V$, we write $\mathcal{P}$ for the set of all simple $s,t$-paths (the common strategy set of the players). We denote $f_P$ as the number of players using path $P \in \mathcal{P}$. We call $f$ a feasible (unsplittable) flow if $\sum_{P \in \mathcal{P}} f_P = N$, and with $f_e$ we denote the number of players using edge $e \in E$, that is, $f_e = \sum_{P \in \mathcal{P} : e \in P} f_P$.

We use the following result due to Fotakis \cite{Fotakis2010}.
\begin{lemma}[Fotakis \cite{Fotakis2010}] Let $\Gamma$ be a congestion game with cost functions $d_e$, and let $\Phi$ be an exact potential for $\Gamma$. An acyclic flow $f$ minimizes the potential function $\Phi$ if and only if
$$
\sum_{e : f_e > g_e} (f_e - g_e)d_e(f_e) - \sum_{e:f_e < g_e} (g_e - f_e)d_e(f_e + 1) \leq 0
$$
for every feasible flow $g$. \label{lem:fotakis_sym_net}
\end{lemma}

The following lemma gives inefficiency results for global minima of the potential function $\Phi$ (compared to any feasible flow). Since the local minima of $\Phi$ correspond to the Nash equilibria of the game $\Gamma$, it follows that the global minima of $\Phi$ are Nash equilibria. Furthermore,  Fotakis \cite{Fotakis2010} shows that every Nash equilibrium of a symmetric congestion game on an extension-parallel graph is a global minimum of the potential function $\Phi$. In particular, this means that the ineffiency results in Lemma \ref{lem:fotakis_implication} hold for the price of stability of symmetric network congestion games, and the price of anarchy of symmetric extension-parallel congestion games.

\begin{lemma}  \label{lem:fotakis_implication}
Let $\Gamma$ be a congestion game with cost functions $d_e(x) = a_e (1 + \rho(x - 1))$, and let $\Phi$ be an exact potential for $\Gamma$. Let $f$ be an acyclic flow minimizing the potential function $\Phi$, then
\begin{eqnarray}
C^\sigma(f) &\leq &\sum_{f_e > g_e} a_e\left[ (f_e - 1)(\rho g_e + (\sigma - \rho)f_e) + g_e\right] \nonumber \\
& &  + \sum_{f_e \leq g_e} a_e \left[ (f_e - 1)(\rho g_e + (\sigma - \rho)f_e) + (1+\rho)g_e - \rho f_e \right]. \nonumber 
\end{eqnarray}
Furthermore, if $h(\rho,\sigma) < 1$ satisfies, 
\begin{equation}\label{eq:pos1a}
(x - 1)(\rho y + (\sigma - \rho)x) + y \leq h(\rho,\sigma) \cdot x[1 + \sigma(x-1)] + g(\rho,\sigma) \cdot y[1 + \sigma(y-1)]
\end{equation}
for all non-negative integers $x > y$, and
\begin{equation}\label{eq:pos1b}
(x - 1)(\rho y + (\sigma - \rho)x) + (1+\rho)y - \rho x \leq h(\rho,\sigma) \cdot  x[1 + \sigma(x-1)] + g(\rho,\sigma) \cdot y[1 + \sigma(y-1)]
\end{equation}
for all non-negative integers $x \leq y$, then $C^\sigma(f)/C^\sigma(g) \leq g(\rho,\sigma) / (1 - h(\rho,\sigma))$. 
\end{lemma} 

\begin{proof}
We write $d_e(x) = a_e[1 + \sigma(x - 1)] +  a_e[(\rho - \sigma)(x - 1)]$ in the left summation, and obtain, using Lemma \ref{lem:fotakis_sym_net},
\begin{eqnarray}
\sum_{f_e > g_e} f_e a_e[1 + \sigma(f_e - 1)] &\leq& \sum_{f_e > g_e} a_e\left( g_e[1 + \rho(f_e - 1)] + f_e(\sigma - \rho)(f_e - 1)\right)  \nonumber \\
& & +\sum_{f_e < g_e} (g_e - f_e) a_e(1+\rho f_e). \nonumber 
\end{eqnarray}
Applying the inequality, we find
\begin{eqnarray}
C^\sigma(f) &=& \sum_{f_e > g_e}f_e a_e[1 + \sigma(f_e - 1)] + \sum_{f_e < g_e}f_e a_e[1 + \sigma(f_e - 1)] + \sum_{f_e = g_e}f_e a_e[1 + \sigma(f_e - 1)] \nonumber \\
& \leq & \sum_{f_e > g_e} a_e\left[ (f_e - 1)(\rho g_e + (\sigma - \rho)f_e) + g_e\right] \nonumber \\
& & +\sum_{f_e < g_e} a_e \left[ (f_e - 1)(\rho g_e + (\sigma - \rho)f_e) + (1+\rho)g_e - \rho f_e \right]+ \sum_{f_e = g_e}f_e a_e[1 + \sigma(f_e - 1)] \nonumber \\
& = & \sum_{f_e > g_e} a_e\left[ (f_e - 1)(\rho g_e + (\sigma - \rho)f_e) + g_e\right] \nonumber \\
& & +\sum_{f_e \leq g_e} a_e \left[ (f_e - 1)(\rho g_e + (\sigma - \rho)f_e) + (1+\rho)g_e - \rho f_e \right] \nonumber
\end{eqnarray}
This completes the proof. \qed
\end{proof}

We continue the proof of the upper bounds in Theorem \ref{thm:pos_upper_network} by showing the result in the statement for $0 < \rho \leq 1$, that is, we define 
$$
h(\rho,1) = 1 - \rho + \rho^2/4  \ \ \ \ \text{ and } \ \ \ \ g(\rho,1) = 1
$$
and prove the correctness of the resulting inequalities in (\ref{eq:pos1a}) and (\ref{eq:pos1b}) (see Lemma \ref{lem:rho01}). The cases $1 \leq \rho \leq 2$ and $2 \leq \rho \leq \infty$ follow (indirectly) from Caragiannis et al. \cite{Caragiannis2010Altruism}.\footnote{The model of Carigiannis et al. \cite{Caragiannis2010Altruism} is equivalent to our model under the transformation $\rho = 1/(1 - \zeta)$, where $\zeta$ is the model parameter of \cite{Caragiannis2010Altruism}. That is, the range $1 \leq \rho \leq 2$ corresponds to $\zeta \in [0,1/2]$, and the range $2 \leq \rho \leq \infty$ to $\zeta \in [1/2,1)$.} The authors use a similar approach as here, but only show the inequality in Lemma \ref{lem:fotakis_sym_net} for Nash equilibria of symmetric singleton congestion games. Nevertheless, the remainder of the analysis carries over  to our model.

\begin{lemma}\label{lem:rho01}
For any integers $x,y \geq 0$ and any $\rho \in (0,1]$ we have, when $x < y$,
$$
\rho\cdot xy + (y - x) \leq \frac{\rho^2}{4}x^2 + y^2,
$$
and, when $x \geq y$,
$$
\rho \cdot xy + (1 - \rho)(y - x) + (1- \rho)x^2 \leq \left(1 - \rho + \frac{\rho^2}{4}\right)x^2 + y^2.
$$
\end{lemma}
\begin{proof}
Let $y = x + z$, where $z$ is a positive integer. Then we have
\begin{eqnarray}
f(x,y) &= &\frac{\rho^2}{4}x^2 + y^2 - \rho xy - (y - x) \nonumber \\ 
&=& \frac{\rho^2}{4}x^2 + (x^2 + 2xz + z^2) - \rho x(x+z) - z \nonumber \\ 
&=& \left(\frac{\rho^2}{4} + 1 - \rho\right)x^2 + \left(2 - \rho\right)xz + z(z-1) \nonumber \\
&\geq& 0, \nonumber
\end{eqnarray}
since $\rho \in (0,1]$, $x \geq 0$ and $z > 0$. \\
\indent For the second inequality, it suffices to show that
$$
g(x,y) = \frac{\rho^2}{4}x^2 + y^2 - \rho xy  - (\rho - 1)(x - y) \geq 0
$$
which can be seen by leaving out the term $(1-\rho)x^2$ on both sides of the inequality.
We first treat the case $y = 0$. Then 
$$
g(x,0) = \frac{\rho^2}{4}x^2 + (1-\rho)x \geq 0
$$
since $x \geq 0$ and $\rho \in (0,1]$. For $y \geq 1$, we write $a = x/y$ (for sake of notation). We have
\begin{eqnarray}
g(x,y)&= &\frac{\rho^2}{4}x^2 + y^2 - \rho xy - (\rho - 1)(x - y)  \nonumber \\ 
&=& \frac{\rho^2}{4}a^2y^2 + y^2 - \rho a y^2 + (1- \rho)(ay - y) \nonumber \\ 
&=& \left[\frac{(\rho a)^2}{4} - \rho a + 1\right]y^2 + (1 - \rho)(a - 1)y \nonumber \\
&=& \left[\frac{\rho a}{2} -1\right]^2y^2 +  (1-\rho)(a-1)y \nonumber \\
&=& \left[\frac{\rho x}{2y} -1\right]^2y^2 +  (1-\rho)\left(\frac{x}{y}-1\right)y \nonumber \\
&\geq& 0, \nonumber
\end{eqnarray}
since $\rho \in (0,1]$, $y \geq 1$ and $a \geq 1$. \qed
\end{proof}

It remains to show tightness of the resulting bounds. For $1 \leq \rho \leq 2$, consider an instance with two players and two resources with resp. cost functions $c_1(x) = x$ and $c_2(x) = (1 + \rho + \epsilon) x$ where $0 < \epsilon \ll \rho$. Then the unique Nash equilibrium is given by $(x_1,x_2) = (2,0)$, and the social optimum by $(x_1^*,x_2^*) = (1,1)$. Sending $\epsilon \rightarrow 0$ gives the desired bound of $4/(2+\rho)$.

For $2 \leq \rho \leq \infty$, we can use the same instance as the for $1 \leq \rho \leq 2$, with the only difference that $c_2(x) = (1 + \rho - \epsilon)x$. Then the social optimum is given by $(x_1^*,x_2^*) = (2,0)$ and the unique Nash equilibrium by $(x_1,x_2) = (1,1)$.

For the case $0 < \rho \leq 1$, the lower bound is technically more involved.

\begin{lemma}
For every fixed (rational) $0 < \rho \leq 1$, and $\epsilon > 0$, there exists a symmetric singleton congestion game for which the price of stability is greater than $4/(\rho(4 - \rho)) - \epsilon$.
\end{lemma}
\begin{proof}
 We choose values of $n$ and $i$ so that
$$
\rho = \frac{i}{n-1}
$$
where, without loss of generality, we may assume that $i$ is even. We will construct a congestion game with $n$ agents and $1 + (n - i/2)$ resources. We let $c_0(x) = x$ and for $e \in \{1,\dots,n-i/2\}$ we define 
$$
c_e(x) = [(1-\rho) + \rho n + \epsilon]x = [1 + i + \epsilon]x, \ \ \ \ \text{ with } 0 < \epsilon \ll \rho.
$$
A socially optimal profile $s^*$ is given by $x_0^* = i/2$ and $x_e^* = 1$ for $e \in \{1,\dots,n-i/2\}$, resulting in
\begin{eqnarray}
\text{SC}(s^*) &=& \left(\frac{i}{2}\right)^2 + \left(n - \frac{i}{2}\right)(1 + i + \epsilon) \nonumber \\
& = & \frac{\rho^2}{4}(n - 1)^2 + \left(n - \frac{1}{2}\rho(n-1)\right)[(1-\rho) + \rho n + \epsilon] \nonumber \\
& = & \frac{\rho^2}{4}(n^2 - 2n + 1) + \left( (1 - \frac{\rho}{2} )n + \frac{\rho}{2}\right)[(1-\rho) + \rho n + \epsilon] \nonumber \\
& = & \left[\frac{\rho^2}{4} + \rho(1 - \frac{\rho}{2})\right]n^2 + \left[-\frac{\rho^2}{2} + (1-\rho + \epsilon)(1 - \frac{\rho}{2}) + \frac{\rho^2}{2} \right]n + \frac{\rho^2}{4} + \frac{\rho}{2}(1-\rho + \epsilon) \nonumber \\
&=&\left[\rho\left(1 - \frac{\rho}{4}\right)\right]n^2 + \left[(1-\rho + \epsilon)(1 - \frac{\rho}{2})\right]n + \frac{\rho}{2}\left(1 - \frac{\rho}{2} + \epsilon \right) \nonumber
\end{eqnarray}

The unique Nash equilibrium is given by the strategy profile $s$ for which $x_0 = n$ and $x_e = 0$ for $e \in \{1,\dots,n-i/2\}$, since the perceived cost on resource $e = 0$ is then precisely $(1-\rho) + \rho n$, so no player can strictly improve its (perceived) cost by deviating to one of the other resources, which have cost $(1-\rho) + \rho n + \epsilon$. The social cost of this equilibrium is $n^2$, which implies that
$$
\frac{C^1(s)}{C^1(s^*)} = \frac{n^2}{\left[\rho\left(1 - \frac{\rho}{4}\right)\right]n^2 + \left[(1-\rho)(1 - \frac{\rho}{2} + \epsilon)\right]n + \frac{\rho}{2}\left(1 - \frac{\rho}{2} + \epsilon \right)} \rightarrow \frac{1}{\rho\left(1 - \rho/4\right)} 
$$
as $n \rightarrow \infty$ (note that this also means that $i \rightarrow \infty$ since $\rho$ is fixed). \qed
\end{proof}

\end{document}